\documentclass[a4paper]{article}

\usepackage{fullpage}

\usepackage{amsmath}
\usepackage{amsthm}
\usepackage{amsfonts}
\usepackage{url}  
\usepackage{graphicx}  
\usepackage{booktabs}
\usepackage[all]{xy}
\usepackage{color}
\usepackage{enumerate}

\newcommand{\body}{\mathsf{body}}
\newcommand{\m}[1]{\ensuremath{\mathsf{#1}}}
\newcommand{\subst}[1]{\ensuremath{\left[{#1}\right]}}
\newcommand{\supp}[1]{\ensuremath{\m{supp}(#1)}}
\newcommand{\var}[1]{\ensuremath{\m{var}(#1)}}
\newcommand{\ans}[3]{\ensuremath{\mathcal A(#1,#2,#3)}}
\newcommand{\hans}[3]{\ensuremath{\mathcal H(#1,#2,#3)}}
\newcommand{\sans}[3]{\ensuremath{\mathcal E(#1,#2,#3)}}
\newcommand{\dans}{\ans QD\tau}
\newcommand{\dhans}{\hans QD\tau}
\newcommand{\dsans}{\sans QD\tau}
\newcommand{\Qex}{\ensuremath{Q_E}}
\newcommand{\Piex}{\ensuremath{\Pi_E}}
\newcommand{\cans}[5]{\ensuremath{\tuple{{#1},#2_{#3}}\vdash_{\mathsf{SLD}}\langle{#4,#5}\rangle}}
\newcommand{\dcans}{\cans QD\tau\theta{\wedge_i\alpha_i}}
\newcommand{\tuple}[1]{\ensuremath{\langle{#1}\rangle}}
\newcommand{\gr}[1]{\ensuremath{\mathcal G}_{#1}}
\newcommand{\sch}[1]{\ensuremath{\mathcal S}_{#1}}
\newcommand{\Pid}{\ensuremath{\Pi^\downarrow}}
\newcommand{\NN}{\mathbb N}
\newcommand{\toremove}{\color{red}}

\newtheorem{proposition}{Proposition}
\newtheorem{example}{Example}
\newtheorem{definition}{Definition}
\newtheorem{lemma}[proposition]{Lemma}
\newtheorem{corollary}[proposition]{Corollary}
\newtheorem{theorem}{Theorem}

\newenvironment{newtext}{}{}

\newcommand{\eoe}{\hfill\ensuremath{\triangleleft}}

\begin{document}

\title{Hypothetical answers to continuous queries over data streams}

\author{Lu\'\i s Cruz-Filipe$^1$ \and Graça Gaspar$^2$ \and Isabel Nunes$^2$}
\date{${}^1$ Dept.\ Mathematics and Computer Science, University of Southern Denmark\\
  ${}^2$ LASIGE, Departmento de Informática, Faculdade de Ciências, Universidade de Lisboa}

\maketitle

\begin{abstract}
  Answers to continuous queries over data streams are often delayed until some relevant input arrives through the data stream.
  These delays may turn answers, when they arrive, obsolete to users who sometimes have to make decisions with no help whatsoever.
  Therefore, it can be useful to provide hypothetical answers -- ``given the current information, it is possible that $X$ will become true at time $t$'' -- instead of no information at all.

  In this paper we present a semantics for queries and corresponding answers that covers such hypothetical answers, together with an incremental online algorithm for updating the set of facts that are consistent with the currently available information.
  Our framework also works in a language supporting negation.
\end{abstract}


\maketitle

\section{Introduction}
\label{sec:intro}

Modern-day reasoning systems often have to react to real-time information about the real world provided by e.g.~sensors.
This information is typically conceptualized as a data stream, which is accessed by the reasoning system.
The reasoning tasks associated to data streams -- usually called \emph{continuous queries} -- are expected to run continuously and produce results through another data stream in an online fashion, as new elements arrive.

A data stream is a potentially unbounded sequence of data items generated by an active, uncontrolled data source.
Elements arrive continuously at the system, potentially unordered, and at unpredictable rates.
Thus, reasoning over data streams requires dealing with incomplete or missing data, potentially storing large amounts of data (in case it might be needed to answer future queries), and providing answers in timely fashion -- among other problems, see e.g.~\cite{Babcock2002,DellAglio2017,Stonebraker2005}.

The output stream is normally ordered by time, which implies that the system may have to delay appending some answer because of uncertainty in possible answers relating to earlier time points.
The length of this delay may be unpredictable (\emph{unbound wait}) or infinite, for example if the query uses operators that range over the whole input data stream (\emph{blocking operations}).
In these cases, answers that have been computed may never be output.
An approach to avoid this problem is to restrict the language by forbidding blocking operations~\cite{Ronca2018,Zaniolo2012}.
Another approach uses the concept of reasoning window~\cite{Becketal2015,Ozcep2017}, which bounds the size of the input that can be used for computing each output (either in time units or in number of events).
\begin{newtext}
In our work, we do not impose such restrictions, and we show that we are able to provide useful information even in some instances of blocking queries.
\end{newtext}

In several applications, it is useful to know that some answers are likely to be produced in the future, since there is already some information that might lead to their generation.
This is the case namely in prognosis systems (e.g., medical diagnosis, stock market prediction), where one can prepare for the possibility of something happening.
To this goal, we propose \emph{hypothetical answers}: answers that are supported by information provided by the input stream, but that still depend on other facts being true in the future.
Knowledge about both the facts that support the answer and possible future facts that may make it true gives users the possibility to make timely, informed decisions in contexts where preemptive measures may have to be taken.

Moreover, by giving such hypothetical answers to the user we cope with unbound wait in a constructive way, since the system is no longer ``mute'' while waiting for an answer to become definitive.

Many existing approaches to reasoning with data streams adapt and extend models, languages and techniques used for querying databases and the semantic web~\cite{Arasu2006,Barbieri2009}.
We develop our theory in line with the works of~\cite{Becketal2015,DaoTran2017,Ozcep2017,Ronca2022,Zaniolo2012}, where continuous queries are treated as rules of a logic program consisting of both rules and facts arriving through a data stream.
Our techniques bear some similarites to previous work on abduction (see the discussion on related work in Section~\ref{sec:rw}), but to the best of our knowledge this is the first time that hypothetical reasoning has been considered in the context of continuous query answering.

\paragraph{Contribution}
We present a declarative semantics for queries in Temporal Datalog, where we define the notions of hypothetical and supported answers.
We define an operational semantics based on SLD-resolution, where we compute answers with premises to queries.
We then show that there is a natural connection between the answers computed by the operational semantics and the hypothetical and supported answers defined declaratively.
By refining SLD-resolution, we obtain an incremental online algorithm for maintaining and updating the set of answers that are consistent with the currently available information.
Finally, \begin{newtext}we define Temporal Datalog with negation and show that our results can be extended to this language.\end{newtext}

\paragraph{Publication history}
A preliminary version of this work, not including proofs, was published in conference proceedings~\cite{ourAAAIpaper}.
The current article extends that work with proofs of all results and further examples.
Section~\ref{sec:neg} (negation), which was only a sketch, was completely restructured, and now includes a precise definition of the declarative and operational semantics, as well as full proofs of all results, and several examples.
Several results in this section were not stated previously, including soundness and completeness.

\paragraph{Structure}
Section~\ref{sec:backgr} revisits some fundamental background notions, namely the formalism from~\cite{Motik}, which we extend in this paper.
Section~\ref{sec:awp} introduces our declarative semantics for continuous queries, defining hypothetical and supported answers, and relates these concepts with the standard definitions of answers.
Section~\ref{sec:sld} presents our operational semantics for continuous queries and relates it to the declarative semantics. 
Section~\ref{sec:algorithm} details our online algorithm to compute supported answers incrementally, as input facts arrive through the data stream, and proves it sound and complete.
Section~\ref{sec:neg} extends our framework to negation by failure.
Section~\ref{sec:rw} compares our proposal to similar ones in the literature, discussing its advantages.

\section{Background}
\label{sec:backgr}

\subsection{Continuous queries in Temporal Datalog}
\label{sec:motiks}
We use the framework from~\cite{Motik} to write continuous queries over datastreams, slightly adapting some definitions.
We work in \emph{Temporal Datalog}, the fragment of negation-free Datalog extended with the special temporal sort from~\cite{Chomicki1988}, which is isomorphic to the set of natural numbers equipped with addition with arbitrary constants.
In Section~\ref{sec:neg} we extend this language with negation.

\paragraph{Syntax of Temporal Datalog}
A \emph{vocabulary} consists of constants (numbers or identifiers in lowercase), variables (single uppercase letters) and predicate symbols (identifiers beginning with an uppercase letter).
All these may be indexed if necessary; occurrences of predicates and variables are distinguished by context.
In examples, we use words in sans serif for concrete constants and predicates.

Constants and variables have one of two sorts: \emph{object} or \emph{temporal}.
An \emph{object term} is either an object (constant) or an object variable.
A \emph{time term} is either a natural number (called a \emph{time point} or \emph{temporal constant}), a time variable, or an expression of the form $T+\m{k}$ where $T$ is a time variable and $\m{k}$ is an integer.

Predicates can take at most one temporal parameter, which we assume to be the last one (if present).
A predicate with no temporal parameter is called \emph{rigid}, otherwise it is called \emph{temporal}.
An atom is an expression $P(t_1,\ldots,t_n)$ where $P$ is a predicate and each $t_i$ is a term of the expected sort.
Without loss of generality, we assume that all predicates are temporal.
Rigid predicates can be made temporal by adding a dummy temporal parameter that does not affect their semantics.
In other words, we replace every instance of $P(t_1,\ldots,t_n)$ by $P'(t_1,\ldots,t_n,0)$.
Thus, $P'(t_1,\ldots,t_n,0)$ holds iff $P(t_1,\ldots,t_n)$ holds -- and the values of $P'$ at time points other than $0$ will have no impact on the semantics.

A rule has the form $\wedge_i\alpha_i\rightarrow\alpha$, where $\alpha$ and each $\alpha_i$ are temporal atoms.
Atom $\alpha$ is called the \emph{head} of the rule, and $\wedge_i\alpha_i$ the \emph{body}.
Rules are assumed to be \emph{safe}: each variable in the head must occur in the body.
A \emph{program} is a set of rules.

If a predicate symbol occurs in an atom in the head of a rule with non-empty body, we call that predicate \emph{intensional}, or an IDB predicate.
A predicate that is defined only through rules with empty body is called \emph{extensional}, or an EDB predicate.
An atom is extensional/intensional (or EDB/IDB atom) according to whether its predicate symbol is extensional or intensional.

A term, atom, rule, or program is \emph{ground} if it contains no variables.
We write $\var\alpha$ for the set of variables occurring in an atom, and extend this function homomorphically to rules and sets.
A \emph{fact} is a function-free ground atom.

Rules are instantiated by means of \emph{substitutions}, which are functions mapping variables to terms of the expected sort.
The \emph{support} of a substitution $\theta$ is the set $\supp\theta=\{X\mid\theta(X)\neq X\}$.
We consider only substitutions with finite support, and write $\theta=\subst{X_1:=t_1,\ldots,X_n:=t_n}$ for the substitution mapping each variable $X_i$ to the term $t_i$, and leaving all remaining variables unchanged.
\begin{newtext}We call each expression $X_i:=t_i$ a \emph{binding}.\end{newtext}
A substitution is \emph{ground} if every variable in its support is mapped to a constant.
An \emph{instance} $r'=r\theta$ of a rule $r$ is obtained by simultaneously replacing every variable $X$ in $r$ by $\theta(X)$ and computing any additions of temporal constants.

A \emph{temporal query} is a pair $Q=\tuple{P,\Pi}$ where $\Pi$ is a program and $P$ is an IDB atom in the language underlying $\Pi$.
We do not require $P$ to be ground, and typically the temporal parameter is uninstantiated.\footnote{The most common exception is if $P$ represents a property that does not depend on time, which normally would be written using a rigid predicate. 
  In our framework, this would mean that $P$'s temporal parameter would be instantiated to $0$.}

A \emph{dataset} is an indexed family of sets of EDB facts (\emph{input facts}), each intuitively corresponding to the facts delivered by a data stream to a reasoning system at each time point.
\begin{definition}
  A \emph{dataset} is a family $D=\{D|_\tau\mid\tau\in\NN\}$, where $D|_\tau$ is a set of EDB facts with timestamp $\tau$.
\end{definition}
We call $D|_\tau$ the \emph{$\tau$-slice} of $D$.
For each $\tau$, we also consider $D$'s $\tau$-history $D_\tau=\bigcup\{D|_{\tau'}\mid\tau'\leq\tau\}$.

From these definitions it follows that $D|_\tau=D_\tau\setminus D_{\tau-1}$ for every $\tau$, and that $D_\tau$ contains only facts whose temporal argument is at most $\tau$.
By convention, $D_{-1}=\emptyset$.

\paragraph{Semantics}
The semantics of Temporal Datalog is a variant of the standard semantics based on Herbrand models.
A Herbrand interpretation $I$ for Temporal Datalog is a set of facts.
If $\alpha$ is an atom with no variables, then we define $\bar\alpha$ as the fact obtained from $\alpha$ by evaluating its temporal term.
We say that $I$ satisfies $\alpha$, $I\models\alpha$, if $\bar\alpha\in I$.
The extension of the notion of satisfaction to the whole language follows the standard construction, and the definition of entailment is the standard one.

An \emph{answer} to a query $Q=\tuple{P,\Pi}$ over a dataset $D$ is a ground substitution $\theta$ whose domain is the set of variables in $P$, satisfying $\Pi\cup D\models P\theta$.
In the context of continuous query answering, we are interested in the case where $D$ is a $\tau$-history of some data stream, which changes with time.
We denote the set of all answers to $Q$ over $D_\tau$ as $\dans$.

We illustrate the extension we propose with a Temporal Datalog program, which is a small variant of Example 1 in~\cite{Motik}.
We will return to this example throughout our presentation.

\begin{example}
  \label{ex:toy}
  A set of wind turbines are scattered throughout the North Sea.
  Each turbine is equipped with a sensor that continuously sends temperature readings $\m{Temp}(\mathit{Device},\mathit{Level},\mathit{Time})$ to a data centre.
  The data centre tracks activation of cooling measures in each turbine, recording malfunctions and shutdowns by means of the following program $\Piex$.
  \begin{align*}
    \m{Temp}(X,\m{high},T) &\rightarrow \m{Flag}(X,T) \\
    \m{Flag}(X,T) \land \m{Flag}(X, T+1) &\rightarrow \m{Cool}(X,T+1) \\
    \m{Cool}(X,T) \land \m{Flag}(X, T+1) &\rightarrow \m{Shdn}(X,T+1) \\
    \m{Shdn}(X,T) &\rightarrow \m{Malf}(X,T-2)
  \end{align*}

  Consider the query $\Qex=\tuple{\m{Malf}(X,T),\Piex}$.
  If $D_0=\{\m{Temp}(\m{wt25},\m{high},0)\}$, then at time point $0$ \begin{newtext}we can infer $\m{Flag(wt25,0)}$, but\end{newtext} there is no output for $\Qex$.
  If $\m{Temp}(\m{wt25},\m{high},1)$ arrives to $D$ at time point $1$, then $D_1=D_0\cup\{\m{Temp}(\m{wt25},\m{high},1)\}$, and \begin{newtext}we can infer $\m{Flag(wt25,1)}$ and $\m{Cool(wt25,1)}$, but\end{newtext} there still is no answer to $\Qex$.
  Finally, the arrival of $\m{Temp}(\m{wt25},\m{high},2)$ to $D$ at time point $2$ yields $D_2=D_1\cup\{\m{Temp}(\m{wt25},\m{high},2)\}$, allowing us to infer \begin{newtext}$\m{Flag(wt25,2)}$, $\m{Cool(wt25,2)}$, $\m{Shdn(wt25,2)}$ and\end{newtext} $\m{Malf}(\m{wt25},0)$.
  Then $\{X:=\m{wt25},T:=0\}\in\ans\Qex D2$.
  \eoe
\end{example}

\subsection{SLD-resolution}
\label{sec:SLD}

We now revisit some concepts from SLD-resolution, adapted from~\cite{BenAri2012,Lloyd1984}.

A \emph{literal} is an atom or its negation.
Atoms are also called \emph{positive} literals, and a negated atom is a \emph{negative} literal.
A \emph{definite clause} is a disjunction of literals containing at most one positive literal.
In the case where all literals are negative, the clause is a \emph{goal}.
We use the standard rule notation for writing definite clauses.

\begin{definition}
  Given two substitutions $\theta=\subst{X_1:=t_1,\ldots,X_m:=t_m}$ and $\sigma=\subst{Y_1:=u_1,\ldots,Y_n:=u_n}$, their \emph{composition} is the substitution $\theta\sigma$ obtained from $$\subst{X_1:=t_1\sigma,\ldots,X_m:=t_m\sigma,Y_1:=u_1,\ldots,Y_n:=u_n}$$ by (i) deleting any binding where $t_i\sigma=X_i$ and (ii) deleting any binding $Y_j:=u_j$ where $Y_j\in\{X_1,\ldots,X_m\}$.
\end{definition}
\begin{newtext}%
Condition~(i) removes redundant identity bindings.
Condition~(ii) ensures that,
\end{newtext}%
for every atom $\alpha$, $\alpha(\theta\sigma)=(\alpha\theta)\sigma$.

\begin{definition}
  Two atomic formulas $P(\vec X)$ and $P(\vec Y)$ are \emph{unifiable} if there exists a substitution $\theta$, also called a \emph{unifier}, such that $P(\vec X)\theta=P(\vec Y)\theta$.

  A unifier $\theta$ of $P(\vec X)$ and $P(\vec Y)$ is called a \emph{most general unifier (mgu)} if for each unifier $\sigma$ of $P(\vec X)$ and $P(\vec Y)$ there exists a substitution $\gamma$ such that $\sigma=\theta\gamma$.
\end{definition}
It is well known that there always exist several mgus of any two unifiable atoms, and that they are unique up to renaming of variables.

Recall that a \emph{goal} is a clause of the form $\neg\wedge_j\beta_j$.
If $C$ is a rule $\wedge_i\alpha_i\to\alpha$, $G$ is a goal $\neg\wedge_j\beta_j$ with $\var G\cap\var C=\emptyset$, and $\theta$ is an mgu of $\alpha$ and $\beta_k$, then the \emph{resolvent} of $G$ and $C$ using $\theta$ is the goal $\neg\left(\bigwedge_{j<k}\beta_j\wedge\bigwedge_i\alpha_i\wedge\bigwedge_{j>k}\beta_j\right)\theta$.

If $\Pi$ is a program and $G$ is a goal, an \emph{SLD-derivation} of $\Pi\cup\{G\}$ is a (finite or infinite) sequence $G_0,G_1,\ldots$ of goals with $G=G_0$, a sequence $C_1,C_2,\ldots$ of $\alpha$-renamings\footnote{\begin{newtext}I.e.~program clauses from $\Pi$ whose variables have been renamed.\end{newtext}} of program clauses of $\Pi$ and a sequence $\theta_1,\theta_2,\ldots$ of substitutions such that $G_{i+1}$ is the resolvent of $G_i$ and $C_{i+1}$ using $\theta_{i+1}$.
A finite SLD-derivation of $\Pi\cup\{G\}$ where the last goal \begin{newtext}$G_n$\end{newtext} is a contradiction ($\Box$) is called an \emph{SLD-refutation} of $\Pi\cup\{G\}$ of length $n$, and the substitution obtained by restricting the composition of $\theta_1,\ldots,\theta_n$ to the variables occurring in $G$ is called a \emph{computed answer} of $\Pi\cup\{G\}$.

\begin{proposition}[Soundness]
  Let $\Pi$ be a program, $G$ be a goal, and suppose that there is an SLD-refutation of $\Pi\cup\{G\}$ with computed answer $\theta$.
  Then $\Pi\models\forall(\neg G\theta)$.
\end{proposition}

\begin{proposition}[Completeness]
  Let $\Pi$ be a program, $G$ be a goal, and suppose that $\Pi\models\forall(\neg G\theta)$.
  Then there is an SLD-refutation of $\Pi\cup\{G\}$ with computed answer $\theta$.
\end{proposition}

\begin{newtext}
The notation $\forall(\neg G\theta)$ stands for the formula $\forall x_1\ldots x_n.\neg G\theta$, where $\var{G\theta}=\{x_1,\ldots,x_n\}$.
Since $G$ is of the form $\neg\wedge_j\beta_j$, this formula can be written equivalently as $\forall x_1\ldots x_n.\wedge_j(\beta_j\theta)$ -- i.e., $\Pi\models\forall(\neg G\theta)$ in the previous two theorems states that $\Pi$ entails any instance of any atom in $G$.
\end{newtext}

There are two degrees of freedom when trying to build SLD-derivations: which literal to choose in the goal at any given point, and which clause to unify it with (in case there are several possibilities).
The choice of the goal is dictated by a \emph{computation rule}; for a given computation rule $R$, the \emph{SLD-tree} for $\Pi$ and $G$ is the tree containing $G$, and where each node contains one descendent for each possible resolvent of the literal chosen by $R$ in that node and the head of a
rule in $\Pi$.

\begin{proposition}[Independence of the computation rule]
  Let $\Pi$ be a program and $G$ be a goal.
  Then every SLD-tree for $\Pi$ and $G$ has the same (finite or infinite) number of branches ending with the empty clause.
\end{proposition}
This result states that the computation rule does not affect whether one can find an SLD-refutation for $\Pi$ and $G$.

\section{Hypothetical answers}
\label{sec:awp}

In our running example, $\m{Temp(wt25,high,0)}$ being delivered at time instant $0$ yields some evidence that $\m{Malf(wt25,0)}$ may turn out to be true.
Later, we may receive further evidence as in the example (the arrival of $\m{Temp(wt25,high,1)}$), or we might find out that this fact will not be true (if $\m{Temp(wt25,high,1)}$ does not arrive).

We propose a theory that also allows for such \emph{hypothetical answers} to a continuous query: if some substitution can become an answer as long as some facts in the future are true, then we output this information.
In this way we can lessen the negative effects of unbound wait.
Hypothetical answers can also refer to future time points: in our example, \subst{X:=\m{wt25},T:=2} would also be output at time point 0 as a substitution that may prove to be an actual answer to the query \tuple{\m{Shdn}(X,T),\Piex} when further information arrives.
 
Our formalism uses ideas from multi-valued logic, where some substitutions correspond to answers (true), others are known not to be answers (false), and others are consistent with the available data, but can not yet be shown to be true or false.
In our example, the fact $\m{Malf(wt25,0)}$ is consistent with the data at time point $0$, and thus ``possible''; it is also consistent with the data at time point $1$, and thus ``more possible''; and it finally becomes (known to be) true at time point 2.

As already motivated, we want answers to give us not only the substitutions that make the query goal true, but also ones that make the query goal possible in the following sense: they depend both on facts that have been delivered by the data stream and on facts that still may be delivered by the data stream (i.e., those whose timestamp is greater than $\tau$).

\begin{definition}
  A ground atom $P(t_1,\ldots,t_n)$ is \emph{future-possible for $\tau$} if $\tau<t_n$.
\end{definition}

\begin{definition}
  \label{defn:hypothetical}
  A \emph{hypothetical answer} to query \begin{newtext}$Q=\tuple{P,\Pi}$\end{newtext} over $D_\tau$ is a pair $\tuple{\theta,H}$, where $\theta$ is a substitution and $H$ is a finite set of ground EDB atoms (the hypotheses) such that:
  \begin{itemize}
  \item $\supp\theta=\var P$;
  \item $H$ only contains atoms future-possible for $\tau$;
  \item $\Pi\cup D_\tau\cup H\models P\theta$;
  \item $H$ is minimal with respect to set inclusion.
  \end{itemize}
  $\dhans$ is the set of hypothetical answers to $Q$ over $D_\tau$.
\end{definition}

Intuitively, a hypothetical answer \tuple{\theta,H} states that $P\theta$ holds if all facts in $H$ are ever delivered by the data stream.
Thus, $P\theta$ is currently backed up by the information available.
In particular, if $H=\emptyset$ then $P\theta$ is an answer in the standard sense (it is a known fact).

\begin{proposition}
  \label{prop:HA-answer}
  Let $Q=\tuple{P,\Pi}$ be a query, $D$ be a data stream and $\tau$ be a time instant.
  If $\tuple{\theta,\emptyset}\in\dhans$, then $\theta\in\dans$.
\end{proposition}
\begin{proof}
  If $\tuple{\theta,H}\in\dhans$, then $\Pi\cup D_\tau\cup H\models P\theta$.
  When $H=\emptyset$, this reduces to $\Pi\cup D_\tau\models P\theta$, which coincides with the definition of answer.
\end{proof}

We can generalize this proposition, formalizing the intuition we gave for the definition of hypothetical answer.

\begin{proposition}
  \label{prop:HA-char}
  Let $Q=\tuple{P,\Pi}$ be a query, $D$ be a data stream and $\tau$ be a time instant.
  If $\tuple{\theta,H}\in\dhans$, then there exist a time point $\tau'\geq\tau$ and a data stream $D'$ such that $D_\tau=D'_\tau$ and $\theta\in\ans Q{D'}{\tau'}$.
\end{proposition}
\begin{proof}
  Let $D'$ be the data stream such that:
  \begin{itemize}
  \item $D'|_t=D|_t$ for $t\leq\tau$;
  \item $D'|_t$ contains the elements of $H$ with timestamp $t$ for $t>\tau$.
  \end{itemize}
  Let $\tau'$ be the highest timestamp in $H$.
  It is straightforward to verify that $D'$ satisfies the thesis.
\end{proof}

\begin{example}
  \label{ex:HA}
  We illustrate these concepts in the context of Example~\ref{ex:toy}.
  
  Consider $\theta=\subst{X:=\m{wt25},T:=0}$.
  Since $\m{Temp}(\m{wt25},\m{high},0)\in D|_0$,
  Then $$\tuple{\theta,\{\m{Temp}(\m{wt25},\m{high},1),\m{Temp}(\m{wt25},\m{high},2)\}}\in\hans\Qex D0\,.$$
  Since $\m{Temp}(\m{wt25},\m{high},1)\in D|_1$, we conclude that $\tuple{\theta,\{\m{Temp}(\m{wt25},\m{high},2)\}}\in\hans\Qex D1$.
  Finally, $\m{Temp}(\m{wt25},\m{high},2)\in D|_2$, so $\tuple{\theta,\emptyset}\in\hans\Qex D2$.
  This answer has no hypotheses, and indeed $\theta\in\ans\Qex D2$.

  Now consider $\theta'=\subst{X:=\m{wt42},T:=1}$ for another constant $\m{wt42}$.
  Then $\hans\Qex D0$ also includes e.g. $\tuple{\theta',\{\m{Temp}(\m{wt42},\m{high},1),\m{Temp}(\m{wt42},\m{high},2),\m{Temp}(\m{wt42},\m{high},3)\}}$.
  However, since $D|_1$ does not contain $\m{Temp}(\m{wt42},\m{high},1)$, there is no element $\tuple{\theta',H'}\in\hans\Qex D\tau$ for $\tau\geq1$.
  \eoe
\end{example}

Hypothetical answers $\tuple{\theta,H}\in\dhans$ where $H\neq\emptyset$ can be further split into two kinds: those that are supported by some fact(s) delivered by the datastream, and those for which there is no evidence whatsover -- they only depend on facts whose truth value is unknown.
For the former, $\Pi\cup H\not\models P\theta$: they rely on some fact from $D_\tau$.
This is the class of answers that interests us, as there is non-trivial information in saying that they may become true.

\begin{definition}
  \label{def:supported}
  Let $Q=\tuple{P,\Pi}$ be a query, $D$ be a data stream and $\tau$ be a time instant.
  A non-empty set of ground EDB atoms $E\subseteq D_\tau$ is \emph{evidence} supporting $\tuple{\theta,H}\in\dhans$ if $E$ is a minimal set satisfying $\Pi\cup E\cup H\models P\theta$.
  A \emph{supported answer} to $Q$ over $D_\tau$ is a triple \tuple{\theta,H,E} \begin{newtext}such that $\Pi\cup H\not\models P\theta$\end{newtext} where $E$ is evidence supporting \tuple{\theta,H}.
  We denote the set of supported answers to $Q$ over $D_\tau$ by $\dsans$.
\end{definition}

Since \begin{newtext}we are working with finite sets,\end{newtext} if $\tuple{\theta,H}\in\dhans$ and $\Pi\cup E\cup H\models P\theta$, then there exists a set $E'$ such that \tuple{\theta,H,E'} is a supported answer to $Q$ over $D_\tau$.
However, in general, several such sets $E'$ may exist.
As a consequence, Propositions~\ref{prop:HA-answer} and~\ref{prop:HA-char} generalize to supported answers in the obvious way.
\begin{newtext}%
The condition $\Pi\cup H\not\models P\theta$ ensures that the evidence must be used to infer the answer.
\end{newtext}%
\begin{example}
  \label{ex:SA}
  In Example~\ref{ex:HA}, the hypothetical answer \tuple{\theta,\{\m{Temp}(\m{wt25},\m{high},1),\m{Temp}(\m{wt25},\m{high},2)\}} is supported by the evidence $\{\m{Temp}(\m{wt25},\m{high},0)\}$, while \tuple{\theta,\{\m{Temp}(\m{wt25},\m{high},2)\}} is supported by $\{\m{Temp}(\m{wt25},\m{high},0),$ $\m{Temp}(\m{wt25},\m{high},1)\}$.

  Since there is no evidence for \tuple{\theta',\{\m{Temp}(\m{wt42},\m{high},1),\m{Temp}(\m{wt42},\m{high},2),\m{Temp}(\m{wt42},\m{high},3)\}}, this answer is not supported.
  \eoe
\end{example}

This example illustrates that unsupported hypothetical answers are not very informative: it is the existence of supporting evidence that distinguishes interesting hypothetical answers from any arbitrary future fact.

However, it is useful to consider even unsupported hypothetical answers in order to develop incremental algorithms to compute supported answers: the sequence of sets $\Theta^E_\tau=\{\theta\mid\tuple{\theta,H,E}\in\dsans\mbox{ for some $H,E$}\}$ is not monotonic, as at every time point new unsupported hypothetical answers may get evidence and supported hypothetical answers may get rejected.
The sequence $\Theta^H_\tau=\{\theta\mid\tuple{\theta,H}\in\dhans\mbox{ for some $H$}\}$, on the other hand, is anti-monotonic, as the following results state.

\begin{proposition}
  \label{prop:HA-split}
  Let $Q=\tuple{P,\Pi}$ be a query, $D$ be a data stream and $\tau$ be a time instant.
  If $\tuple{\theta,H}\in\dhans$, then there exists $H^0$ such that $\tuple{\theta,H^0}\in\hans QD{-1}$ and $H=H^0\setminus D_\tau$.
  Furthermore, if $H\neq H^0$, then $\tuple{\theta,H,H^0\setminus H}\in\dsans$.
\end{proposition}
\begin{proof}
  Recall that $D_{-1}=\emptyset$ by convention.
  If $\tuple{\theta,H}\in\dhans$, then $\Pi\cup D_\tau\cup H\models P\theta$.
  Since $D_\tau\cup H$ is finite, there is a minimal subset $H^0$ of $D_\tau\cup H$ with the property that $H\subseteq H^0$ and $\Pi\cup H^0\models P\theta$.
  Clearly $H=H^0\setminus D_\tau$.

  Assume that $H^-\subseteq H^0$ is also such that $\Pi\cup H^-\models P\theta$.
  Then $\Pi\cup D_\tau\cup(H^-\setminus D_\tau)\models P\theta$; but $\tuple{\theta,H}\in\dhans$, so $H\subseteq H^-\setminus D_\tau$ and therefore also $H\subseteq H^-$.
  By definition of $H^0$, this implies that $H^0\subseteq H^-$, hence $\tuple{\theta,H^0}\in\hans QD{-1}$.

  Finally, if $E\subseteq H^0\setminus H$ is evidence supporting \tuple{\theta,H}, then $\Pi\cup E\cup H\models P\theta$, hence $H^0\subseteq E\cup H$, since $\tuple{\theta,E\cup H}\in\hans QD{-1}$.
\end{proof}

\begin{proposition}
  \label{prop:HA-split-general}
  Let $Q=\tuple{P,\Pi}$ be a query, $D$ be a data stream and $\tau$ be a time instant.
  If $\tuple{\theta,H}\in\dhans$ and $\tau'<\tau$, then there exists $\tuple{\theta,H'}\in\hans QD{\tau'}$ such that $H=H'\setminus\left(D_\tau\setminus D_{\tau'}\right)$.
\end{proposition}
\begin{proof}
  Just as the proof of the previous proposition, but dividing $\Pi\cup D_\tau$ into   $\Pi\cup D_{\tau'}$ and $D_\tau\setminus D_{\tau'}$ instead of into $\Pi$ and $D_\tau$.
\end{proof}

Examples~\ref{ex:HA} and~\ref{ex:SA} also illustrate this property, with hypotheses turning into evidence as time progresses.
Since $D_{-1}=\emptyset$, Proposition~\ref{prop:HA-split} is a particular case of Proposition~\ref{prop:HA-split-general}.

In the next sections we show how to compute hypothetical answers and the corresponding sets of evidence for a given continuous query.

\section{Operational semantics via SLD-resolution}
\label{sec:sld}

The definitions of hypothetical and supported answers are declarative.
We now show how SLD-resolution can be adapted to algorithms that compute these answers.
We use standard results about SLD-resolution, see for example~\cite{Lloyd1984}.

We begin with a simple observation: since the only function symbol in our language is addition of
temporal parameters (which is invertible), we can always choose mgus that do not replace variables
in the goal with new ones.

\begin{lemma}
  \label{lem:mgu-vars}
  Let $\neg\wedge_i\alpha_i$ be a goal and $\wedge_j\beta_j\to\beta$ be a rule such that
  $\beta$ is unifiable with $\alpha_k$ for some $k$.
  Then there is an mgu $\theta=\subst{X_1:=t_1,\ldots,X_n:=t_n}$ of $\alpha_k$ and $\beta$ such that all variables occurring in $t_1,\ldots,t_n$ also occur in $\alpha_k$.
\end{lemma}
\begin{proof}
  Let $\rho=\subst{X_1:=t_1,\ldots,X_n:=t_n}$ be an mgu of $\alpha_k$ and $\beta$.
  For each $i$, $t_i$ can either be a variable $Y_i$ or a time expression $T_i+k_i$.
  First, iteratively build a substitution $\sigma$ as follows: for $i\in[1,\ldots,n]$, if $X_i$ occurs in $\alpha_k$ but $t_i$ does not and $\sigma$ does not yet include a replacement for the variable in $t_i$, extend $\sigma$ with $Y_i:=X_i$, if $t_i$ is $Y_i$, or $T_i:=X_i-k_i$, if $t_i$ is $T_i+k_i$.

  We now show that $\theta=\rho\sigma$ is an mgu of $\alpha_k$ and $\beta$ with the desired property.
  If $X:=t\in\theta$, then either (i)~$X$ is $X_i$ and $t$ is $t_i\sigma$ for some $i$ or (ii)~$X:=t\in\sigma$.
  In case~(i), by construction of $\sigma$ if $X_i$ occurs in $\alpha_k$ but $t_i$ includes a variable not in $\alpha_k$, then $\sigma$ replaces that variable with a term using only variables in $\alpha_k$. In case~(ii), by construction $X$ does not occur in $\alpha_k$.

  To show that $\theta$ is an mgu of $\alpha_k$ and $\beta$ it suffices to observe that $\sigma$ is invertible, with $\sigma^{-1}=\subst{X:=Y \mid Y:=X\in\sigma}\subst{X:=T-k\mid T:=X+k\in\sigma}$.
\end{proof}

Without loss of generality, we assume all mgus in SLD-derivations to have the property in Lemma~\ref{lem:mgu-vars}.

Our intuition is as follows: classical SLD-resolution computes substitutions that make a conjunction of atoms a logical consequence of a program by constructing SLD-derivations that end in the empty clause.
We relax this by allowing derivations to end with a goal if: this goal only refers to EDB predicates and all the temporal terms in it refer to future instants (possibly after further instantiation).
This makes the notion of derivation also dependent on a time parameter.

\begin{definition}
  An atom $P(t_1,\ldots,t_n)$ is a \emph{potentially future atom wrt $\tau$} if it is
  future-possible for $\tau$ or either $t_n$ contains a temporal variable.
\end{definition}
This concept generalizes the notion of future-possible atom for $\tau$ to possibly non-ground atoms.
In particular, any atom whose time parameter contains a variable is a potentially future atom, since
it may be instantiated to a time point in the future.

\begin{definition}
  An \emph{SLD-refutation with future premises} of \begin{newtext}$Q=\tuple{P,\Pi}$\end{newtext} over $D_\tau$ is a finite SLD-derivation of $\Pi\cup D_\tau\cup\{\neg P\}$ whose last goal only contains potentially future EDB atoms wrt $\tau$.

  If $\mathcal D$ is an SLD-refutation with future premises of $Q$ over $D_\tau$ with last goal $G=\neg\wedge_i\alpha_i$ and $\theta$ is the substitution obtained by restricting the composition of the mgus in $\mathcal D$ to $\var P$, then \tuple{\theta,\wedge_i\alpha_i} is a \emph{computed answer with premises} to $Q$ over $D_\tau$, denoted \dcans.
\end{definition}

\begin{example}
  Consider query $\Qex$ from Example~\ref{ex:toy} and $\tau=1$.
  There is an SLD-derivation of $\begin{newtext}\Pi_E\end{newtext}\cup D_1\cup\{\neg\m{Malf}(X,T)\}$ ending with $\leftarrow\m{Temp}(\m{wt25},\m{high},2)$, which contains a single potentially future EDB atom with respect to~$1$.
  Thus, $\cans\Qex D1\theta{\m{Temp}(\m{wt25},\m{high},2)}$ with $\theta={\subst{X:=\m{wt25},T:=0}}$.
  \eoe
\end{example}

As this example illustrates, computed answers with premises are the operational counterpart to hypothetical answers, with two caveats.
First, a computed answer with premises need not be ground: there may be some universally quantified variables in the last goal.
Second, $\wedge_i\alpha_i$ may contain redundant conjuncts, in the sense that they might not be needed to establish the goal.
We briefly illustrate these two features.

\begin{example}
  \label{ex:SLD-problems}
  In our running example, there is also an SLD-derivation of $\begin{newtext}\Pi_E\end{newtext}\cup D_1\cup\{\neg\m{Malf}(X,T)\}$ ending with the goal $\neg\bigwedge_{i=0}^2\m{Temp}(X,\m{high},T+i)$, which only contains potentially future EDB atoms with respect to~$1$.
  Thus also $\cans{Q_E}D1\emptyset{\bigwedge_{i=0}^2\m{Temp}(X,\m{high},T+i)}$.
  \eoe
\end{example}

\begin{example}
  Consider the following program $\Pi'$:
  \begin{align*}
    \m{P}(\m{a},T) &\rightarrow \m{R}(\m{a},T) &
    \m{P}(\m{a},T)\wedge\m{Q}(\m{a},T) &\rightarrow \m{R}(\m{a},T)
  \end{align*}
  and the query $Q'=\tuple{\m{R}(X,T),\Pi'}$.
  
  Let $D'|_0=\emptyset$.
  Then there is an SLD-derivation of $\Pi'\cup D'_0\cup\{\neg\m{R}(X,T)\}$ that ends with the goal $\neg\left(\m{P}(\m{a},T)\wedge\m{Q}(\m{a},T)\right)$, which only contains potentially future EDB atoms wrt $\tau$.
  Thus $\cans{Q'}{D'}0{\subst{X:=\m{a}}}{\m{P}(\m{a},T)\wedge\m{Q}(\m{a},T)}$.
  However, atom $\m{Q}(\m{a},T)$ is redundant, since $\m{P}(\m{a},T)$ suffices to make $\subst{X:=\m{a}}$ an answer to $Q$ for any~$T$.
  
  (Observe that also $\cans{Q'}{D'}0{\subst{X:=\m{a}}}{\m{P}(\m{a},T)}$, but produced by a different SLD-derivation.)
  \eoe
\end{example}

We now look at the relationship between the operational definition of computed answer with premises and the notion of hypothetical answer.
The examples above show that these notions do not precisely correspond.
However, we can show that computed answers with premises approximate hypothetical answers and, conversely, every hypothetical answer is a grounded instance of a computed answer with premises.

\begin{theorem}[Soundness]
  \label{thm:SLD-sound}
  Let $Q=\tuple{P,\Pi}$ be a query, $D$ be a data stream and $\tau$ be a time instant.
  Assume that $\dcans$.
  Let $\sigma$ be a ground substitution such that: (i)~$\supp\sigma=\var{\wedge_i\alpha_i}\cup(\var P\setminus\supp\theta)$ and (ii)~if $\alpha_i$ is $P(t_1,\ldots,t_n)$, then $t_n\sigma>\tau$.
  Then there is a set $H\subseteq\{\alpha_i\sigma\}_i$ such that $\tuple{(\theta\sigma)|_{\var P},H}\in\dhans$.
\end{theorem}
\begin{proof}
  Assume that there is some SLD-refutation with future premises of $Q$ over $D_\tau$.
  Then this is an SLD-derivation whose last goal $G=\neg\wedge_i\alpha_i$ only contains potentially future EDB atoms with respect to $\tau$.
  Let $\sigma$ be any substitution in the conditions of the hypothesis.
  Taking $H'=\{\alpha_i\sigma\}_i$, we can extend this SLD-derivation to a (standard) SLD-refutation for $\Pi\cup D_\tau\cup H'\cup\{\neg P\}$, by resolving $G$ with each of the $\alpha_i$ in turn.
  The computed answer is then the restriction of $\theta\sigma$ to $\var P$.
  By soundness of SLD-resolution, $\Pi\cup D_\tau\cup H'\models P(\theta\sigma)|_{\var P}$.
  Since \begin{newtext}$H'$ is finite,\end{newtext} we can find a minimal set $H\subseteq H'$ with the latter property.
\end{proof}

\begin{theorem}[Completeness]
  \label{thm:SLD-compl}
  Let $Q=\tuple{P,\Pi}$ be a query, $D$ be a data stream and $\tau$ be a time instant.
  If $\tuple{\theta,H}\in\dhans$, then there exist substitutions $\rho$ and $\sigma$ and a finite set of atoms $\{\alpha_i\}_i$ such that $\theta=\rho\sigma$, $H=\{\alpha_i\sigma\}_i$ and $\cans QD\tau\rho{\wedge_i\alpha_i}$.
\end{theorem}
\begin{proof}
  Suppose $\tuple{\theta,H}\in\dhans$.
  Then $\Pi\cup D_\tau\cup H\models P\theta$.
  By completeness of SLD-resolution, there exist substitutions $\gamma$ and $\delta$ and an SLD-derivation for $\Pi\cup D_\tau\cup H\cup\{\neg P\}$ with computed answer $\gamma$ such that $\theta=\gamma\delta$.

  By minimality of $H$, for each $\alpha\in H$ there must exist a step in this SLD-derivation where the current goal is resolved with $\alpha$.
  By independence of the computation rule, we can assume that these are the last steps in the derivation.
  Let $\mathcal D'$ be the derivation consisting only of these steps, and $\mathcal D$ be the original derivation without $\mathcal D'$.
  Let $\rho$ be the answer computed by $\mathcal D$ and $G=\neg\wedge_i\alpha_i$ be its last goal, let $\rho'$ be the answer computed by $\mathcal D'$, and define $\sigma=\rho'\delta$.
  Then:
  \begin{itemize}
  \item Let $X\in\supp\theta$; then $X$ occurs in $P$.
    If $\rho(X)$ occurs in $G$, then by construction $\gamma(X)=(\rho\rho')(X)$.
    If $\rho(X)$ is a ground term or $\rho(X)$ does not occur in $G$, then trivially $\gamma(X)=\rho(X)=(\rho\rho')(X)$ since $\rho'$ does not change $\rho(X)$.
    In either case, $\theta=\gamma\delta=\rho\rho'\delta=\rho\sigma$.
  \item $H=\{\alpha_i\sigma\}_i$: by construction of $\mathcal D'$, $H=\{\alpha_i\rho'\}_i$, and since $\alpha_i\rho'$ is ground for each $i$, it is also equal to $\alpha_i\rho'\delta=\alpha_i\sigma$.
  \end{itemize}
  The derivation $\mathcal D$ shows that $\cans QD\tau\rho{\wedge_i\alpha_i}$.
\end{proof}

All notions introduced in this section depend on the time parameter $\tau$, and in particular on the history dataset $D_\tau$.
In the next section, we explore the idea of ``organizing'' the SLD-derivation adequately to pre-process $\Pi$ independently of $D_\tau$, so that the computation of (hypothetical) answers can be split into an offline part and a less expensive online part.

\section{Incremental computation of hypothetical answers}
\label{sec:algorithm}

Proposition~\ref{prop:HA-split-general} states that the set of hypothetical answers evolves as time passes, with hypothetical answers either gaining evidence and becoming query answers or being put aside due to their dependence on facts that turn out not to be true.

In this section, we show how we can use this temporal evolution to compute supported answers incrementally.
\begin{newtext}
We start by showing a simple result on SLD-derivations that allows us to reorganize the resolution steps based on the timestamp (\S~\ref{sec:more-sld}).
Using this result, we propose a two-stage process for computing and updating supported answers to continuous queries over data streams: a first (offline) stage pre-processes the program as much as possible (\S~\ref{sec:preproc}), and a second (online) stage updates a set of supported answers based on incoming information (\S~\ref{sec:online}).
Correctness of the online stage, due to its complexity, is delegated to a separate section (\S~\ref{sec:online-corr}).
\end{newtext}

\subsection{A two-stage process}
\label{sec:more-sld}

We start by revisiting SLD-derivations and showing how they can reflect the structure \begin{newtext}of the datastream\end{newtext}.

\begin{proposition}
  \label{prop:SLD-strat}
  Let $Q=\tuple{P,\Pi}$ be a query and $D$ be a data stream.
  For any time constant $\tau$, if $\dcans$, then there exist an SLD-refutation with future premises of $Q$ over $D_\tau$ computing \tuple{\theta,\wedge_i\alpha_i} and a sequence $k_{-1}\leq k_0\leq\ldots\leq k_\tau$ such that:
  \begin{itemize}
  \item goals $G_1,\ldots,G_{k_{-1}}$ are obtained by resolving with clauses from $\Pi$;
  \item for $0\leq i\leq\tau$, goals $G_{k_{i-1}+1},\ldots,G_{k_i}$ are obtained by resolving with clauses from $D|_i$.
  \end{itemize}
\end{proposition}
\begin{proof}
  Immediate consequence of the independence of the computation rule.
\end{proof}

An SLD-refutation with future premises with the property guaranteed by
Proposition~\ref{prop:SLD-strat} is called a \emph{stratified} SLD-refutation with future premises.
Since data stream $D$ only contains EDB atoms, it also follows that in a stratified SLD-refutation all goals after $G_{k_{-1}}$ are always resolved with EDB atoms.
Furthermore, each goal $G_{k_i}$ contains only potentially future EDB atoms with respect to $i$.
Let $\theta_i$ be the restriction of the composition of all substitutions in the SLD-derivation up to step $k_i$ to $\var P$.
Then $G_{k_i}=\neg\wedge_j\alpha_j$ represents all hypothetical answers to $Q$ over $D_i$ of the form \tuple{(\theta_i\sigma)|_{\var P},\wedge_j\alpha_j} for some ground substitution $\sigma$ (cf.~Theorem~\ref{thm:SLD-sound}).

This yields an online procedure to compute supported answers to continuous queries over data streams.
In a pre-processing step, we calculate all computed answers with premises to $Q$ over $D_{-1}$, and keep the ones with minimal set of formulas.
(Note that Theorem~\ref{thm:SLD-compl} guarantees that all minimal sets are generated by this procedure, although some non-minimal sets may also appear as in Example~\ref{ex:SLD-problems}.)
The online part of the procedure then resolves each of these sets with the facts delivered by the data stream, adding the resulting resolvents to a set of schemata of supported answers (i.e.~where variables may still occur).
By Proposition~\ref{prop:SLD-strat}, if there is at least one resolution step at this stage, then the hypothetical answers represented by these schemata all have evidence, so they are indeed supported.

\subsection{Pre-processing step}
\label{sec:preproc}

\begin{newtext}
We now look at the pre-processing step of our procedure in more detail.
\end{newtext}
In general, this step may not terminate, as the following example illustrates.

\begin{example}
  \label{ex:infinite}
  Consider the following program $\Pi''$, where $\m R$ is an extensional predicate and $\m S$ is an intensional predicate.
  \begin{align*}
    \m S(X,T) &\rightarrow \m S(X,T+1) &
    \m R(X,T) & \rightarrow \m S(X,T)
  \end{align*}
  If $\m{R(a,t_0)}$ is delivered by the datastream, then $\m{S(a,}t)$ is true for every $t\geq\m{t_0}$.
  
  Thus, $\tuple{[X:=\m a],\{\m R(\m a,T-k)\}}\in\hans{\tuple{\m S(X,T),\Pi''}}D0$ for all $k$.
  The pre-processing step needs to output this infinite set, so it cannot terminate.
  \eoe
\end{example}

We establish termination of the pre-processing step for two different classes of queries.
A query $Q=\tuple{P,\Pi}$ is \emph{connected} if each rule in $P$ contains at most one temporal variable, which occurs in the head whenever it occurs in the body; and it is \emph{nonrecursive} if the directed graph induced by its dependencies is acyclic, cf.~\cite{Motik}.
\begin{proposition}
  \label{prop:termination}
  Let $Q=\tuple{P,\Pi}$ be a nonrecursive and connected query.
  Then the set of all computed answers with premises to $Q$ over $D_{-1}$ can be computed in finite time.
\end{proposition}
\begin{proof}
  Let $T$ be the (only) temporal variable in $P$.
  \begin{newtext}%
  We first show that all SLD-derivations for $\Pi\cup\neg{P}$ have a maximum depth.
  First, associate to each predicate $P$ the maximum length of a path in the dependency graph for $\Pi$ starting from $P$.
  Next, associate to each goal the sorted sequence of these values for each of its atoms.
  Then each resolution step decreases the sequence assigned to the goal with respect to the lexicographic ordering.
  \end{newtext}
  Since this ordering is well-founded, the SLD-derivation must terminate.

  Furthermore, since $\Pi$ is finite, there is a finite number of possible descendants for each node.
  Therefore, the tree containing all possible SLD-derivations for $\Pi\cup\neg P$ is a finite branching tree with finite height, and by König's Lemma it is finite.

  Since each resolution step terminates (possibly with failure) in finite time, this tree can be built in finite time.
\end{proof}

The algorithm implicit in the proof of Proposition~\ref{prop:termination} can be improved by standard techniques (e.g.~by keeping track of generated nodes to avoid duplicates).
However, since it is a pre-processing step that is done offline and only once, we do not discuss such optimizations.

The authors of \cite{Motik} also formally define query delay\footnote{These are not communication delays, but delays in producing answers to queries.} and window size as follows, where $Q$ is a query with temporal variable $T$.
\begin{definition}
  \label{defn:delay}
  A \emph{delay} for $Q$ is a natural number $d$ such that: for every substitution $\theta$ and every $\tau\geq T\theta+d$, $\theta\in\ans QD\tau$ iff $\theta\in\ans QD{T\theta+d}$.

  \label{defn:window}
  \begin{newtext}
  If $Q$ has a delay $d$, then a natural number $w$ is a \emph{window size} for $Q$ if: for every substitution $\theta$ and all $\tau>\tau'-d$, $\theta\in\ans QD\tau$ iff $\theta\in\ans Q{D\setminus D_{\tau'-w}}\tau$.
  \end{newtext}
\end{definition}

To show termination of the pre-processing step assuming existence of a delay and window size, we apply a slightly modified version of SLD-resolution: we do not allow fresh non-temporal variables to be added.
So if a rule includes new variables in its body, we generate a node for each of its possible instances.
We start by proving an auxiliary lemma.

\begin{lemma}
  \label{lem:termination}
  If $Q=\tuple{P,\Pi}$ has delay $d$ and window size $w$, then there exist two natural numbers $d'$ and $w'$ such that: if an atom in a goal in an SLD-derivation for $\Pi\cup\{\neg P\}$ contains a temporal argument $T+k$ with $k>d'$ or $k<{-w'}$, then the subtree with that goal as root does not contain a leaf including an extensional atom with a temporal parameter dependent on $T$.
\end{lemma}
\begin{proof}
  Assume that the thesis does not hold, i.e.~for any $d'$ and $w'$ there exists an atom in an SLD-derivation for $\Pi\cup\{\neg P\}$ whose subtree contains a leaf including an extensional atom with temporal parameter dependent on $T$.
  By the independence of the computation rule, this property must hold for any SLD-derivation for $\Pi\cup\{\neg P\}$.

  Then there is a predicate symbol that occurs infinitely many times in the tree for $\Pi\cup\{\neg P\}$ with distinct instantiations of the temporal argument.
  Since the number of distinct instantiations for the non-temporal arguments is finite (up to $\alpha$-equivalence), there must be at least one instantiation that occurs infinitely often.
  If applying SLD-resolution to it generates an extensional literal whose temporal argument depends on $T$, then the tree would have infinitely many leaves containing distinct occurrences of its underlying predicate symbol.

  For any $t$, it is then straightforward to construct an instance of the data stream and a substitution $\theta$ such that $\theta\not\in\ans QD{T\theta}$ and $\theta\in\ans QD{T\theta+t}$, so no number smaller than $t$ can be a delay for $Q$.
  Similarly, we can construct an instance of the data stream and a substitution such that $\theta$ is an answer to $Q$ over $D_{\theta+t}$ but not over $D_{\theta+t}\setminus D_\theta$, so $t$ cannot be a window size for $Q$.
\end{proof}

\begin{proposition}
  \label{thm:termination-char}
  If $Q$ has a delay $d$ and a window $w$, then the set of all computed answers with premises to $Q$ over $D_{-1}$ can be computed in finite time.
\end{proposition}
\begin{proof}
  Let $w'$ and $d'$ be the two natural numbers guaranteed to exist by the previous lemma.
  Then, if a goal contains a positive literal with temporal parameter $T+k$ with $k>d'$ or $k<{-w'}$, we know that the subtree generated from that literal cannot yield leaves with temporal parameter depending on $T$.
  Hence, either this tree contains only failed branches, or it contains leaves that do not depend on $T$, or it is infinite.

  Although we do not know the value of $d'$ and $w'$, we can still use this knowledge by performing the construction of the tree in a depth-first fashion.
  Furthermore, we choose a computation rule that delays atoms with temporal variable other than $T$.
  For each node, before expanding it we first check whether the atom we are processing has already appeared before (modulo temporal variables), and if so we skip the construction of the corresponding subtree by directly inserting the corresponding leaves.

  If an atom contains a temporal variable other than $T$, all leaves in its subtree contain no extensional predicates with non-constant temporal argument (otherwise there would not be a delay for $T$).
  Therefore these sub-goals can also be treated uniformly.

  Finally, suppose that a branch is infinite.
  Again using a finiteness argument, any of its branches must necessarily at some point include repeated literals.
  In this case we can immediately mark it as failed.
\end{proof}

\begin{newtext}
In general, given a query $Q=\tuple{P,\Pi}$, we can try to compute a finite set $\mathcal P_Q$ that represents $\hans QD{-1}$ as follows.
As in the proof of the previous proposition, choose a computation rule that delays atoms with temporal variable other than $T$.
Build an SLD-tree with $\neg P$ as initial goal, with the following two extra checks.
\begin{itemize}
\item Before expanding a node, check whether the atom being processed has already appeared before (modulo temporal variables), and if so immediately generate the corresponding leaves.
\item If the chosen atom has already appeared in the derivation, immediately mark the branch as failed.
\end{itemize}

Building the SLD-tree amounts to the worst case of deciding entailment using SLD-resolution, which is NP-complete.

\begin{definition}
  $\mathcal P_Q$ is the set of all entries \tuple{\theta,\{\alpha_i\}_i} such that the SLD-tree for $\neg P$ contains a leaf with computed answer \tuple{\theta,\wedge_i\alpha_i} with premises to $Q$ over $D_{-1}$ such that $\{\alpha_i\}_i$ is minimal in the whole tree.
\end{definition}
\end{newtext}

Each tuple $\tuple{\theta,H}\in\mathcal P_Q$ represents the set of all hypothetical answers \tuple{\theta\sigma,H\sigma} as in Theorem~\ref{thm:SLD-sound}.

\paragraph{Remark.}
The hypotheses in Propositions~\ref{prop:termination} and~\ref{thm:termination-char} are not necessary to guarantee termination of the algorithm presented for the pre-processing step.
Indeed, consider the following example.

\begin{example}
  \label{ex:unbound}
  In the context of our running example, we say that a turbine has a manufacturing defect if it exhibits two specific failures during its lifetime: at some time it overheats, and at some (different) time it does not send a temperature reading.

  Since this is a manufacturing defect, it holds at timepoint $0$, regardless of when the failures actually occur.\footnote{This is actually an example of a rigid predicate, which we model as described.}
  We can model this property by the rule \[\m{Temp}(X,\m{high},T_1),\m{Temp}(X,\m{n/a},T_2) \to \m{Defective}(X,0)\,.\]
  Let $\Pi'_E$ be the program obtained from $\Pi_E$ by adding this rule, and consider now the query $Q'=\tuple{\m{Defective}(X,T),\Pi'_E}$.
  Performing SLD-resolution between $\Pi'_E$ and $\m{Defective}(X,0)$ yields the new goal $\neg\left(\m{Temp}(X,\m{high},T_1)\wedge\m{Temp}(X,\m{n/a},T_2)\right)$, which only contains potentially future atoms with respect to $-1$.
  \eoe
\end{example}
The program in this example includes a rule that uses two different time variables.
As a consequence, the query is not connected, and it does not have a delay or window size (since no single predicate can use both $T_1$ and $T_2$).
Still, the pre-processing step terminates.

We assume from this point onwards that the query $Q$ is connected.

\subsection{The online step}
\label{sec:online}

\begin{newtext}
We now describe how to compute and update the set $\dsans$ in a schematic way, i.e., using variables to represent families of similar hypotheses.

\begin{definition}
  A \emph{schematic supported answer} for a query $Q$ at time $\tau$ is a triple \tuple{\theta,E,H} where $\theta$ is a substitution, $E$ is a set of ground EDB atoms and $H$ is a set of (not necessarily ground) EDB atoms such that $\tuple{\theta,E,H\sigma}\in\dsans$ for every ground substitution $\sigma$.
\end{definition}
\end{newtext}

\begin{definition}
  A substitution $\sigma$ is a \emph{minimal substitution with property $P$} if: for any substitution $\theta$ with property $P$, there exists $\rho$ such that $\theta=\sigma\rho$.
\end{definition}
In particular, an mgu for a set is a minimal substitution unifying all formulas in the set.

\begin{definition}
  \begin{newtext}
  Let $Q$ be a query.
  For each time point $\tau$, we define the set $\sch\tau$ as follows.
  \end{newtext}
  \begin{itemize}
  \item $\sch{-1}=\emptyset$
  \item If $\tuple{\theta,H}\in\mathcal P_Q$ and $\sigma$ is a minimal substitution such that $H\sigma\cap D|_\tau\neq\emptyset$ and $H\sigma\setminus D|_\tau$ only contains potentially future atoms wrt $\tau$, then $\tuple{\theta\sigma,H\sigma\cap D|_\tau,H\sigma\setminus D|_\tau}\in S_\tau$.
  \item If $\tuple{\theta,E,H}\in\sch{\tau-1}$ and $\sigma$ is a minimal substitution such that $H\sigma\setminus D|_\tau$ only contains potentially future atoms wrt $\tau$, then $\tuple{\theta\sigma,E\cup E',H\sigma\setminus D|_\tau}\in\sch\tau$, where $E'=H\sigma\cap D|_\tau$.
  \end{itemize}
\end{definition}

\begin{newtext}
In \S~\ref{sec:online-corr}, we prove that $\sch\tau$ is a set of schematic supported answers that captures the set $\dsans$ precisely.
Before that, we show that $\sch\tau$ can be computed as follows.

\begin{definition}
   Let $S$ be a set of schematic supported answers for a query $Q$ at time $\tau-1$ and $\mathcal P_Q$ be the result of pre-processing $Q$.
   The \emph{update} of $S$ using $\mathcal P_Q$, $\mathcal U(S,\mathcal P_Q)$, is computed as follows.
  \begin{itemize}
  \item For each $\tuple{\theta,H}\in\mathcal P_Q$, let $M\subseteq H$ be the set of atoms with minimal timestamp.
    For each computed answer $\sigma$ to $\neg\bigwedge M$ and $D|_\tau$, add $\tuple{\theta\sigma,M\sigma,(H\setminus M)\sigma}$ to $\mathcal U(S,\mathcal P_Q)$.
  \item For each $\tuple{\theta,E,H}\in S$, let $M\subseteq H$ be the set of atoms with timestamp $\tau$.
    For each computed answer $\sigma$ to $\neg\bigwedge M$ and $D|_\tau$, add $\tuple{\theta\sigma,(E\cup M)\sigma,(H\setminus M)\sigma}$ to $\mathcal U(S,\mathcal P_Q)$.
  \end{itemize}
\end{definition}
\end{newtext}%

\begin{lemma}
  \label{lem:new-lemma}
  \begin{newtext}For all $\tau\geq 0$, $\sch\tau=\mathcal U(\sch{\tau-1},\mathcal P_Q)$.\end{newtext}
\end{lemma}
\begin{proof}
  Both cases
  \begin{newtext}%
  of the computation of $\mathcal U$%
  \end{newtext}
  are similar.
  We show the first one in detail.
  Assume that $\tuple{\theta,H}\in\mathcal P_Q$ and $\sigma$ is a minimal substitution such that $H\sigma\cap D|_\tau\neq\emptyset$ and $H\sigma\setminus D|_\tau$ only contains potentially future atoms wrt $\tau$.
  By definition of potentially future atom and the fact that all the elements of $D|_\tau$ have the same timestamp ($\tau$), the set $H\sigma\cap D|_\tau$ contains the elements of $H\sigma$ with minimal timestamp.
  Therefore $\sigma$ is a minimal substitution unifying all the elements of $M$ with elements of $D|_\tau$, and the thesis follows by soundness and completeness of SLD-resolution.
  Note that $M$ cannot be empty, since $H\neq\emptyset$.

  For the second case, we observe that all elements of $H$ have their timestamp instantiated (by connectedness, there can be at most one temporal variable, which is instantiated when new elements are added from $\mathcal P_Q$).
  The proof then proceeds as in the previous case.
\end{proof}

\begin{example}
  \label{ex:blahblah}
  We illustrate this mechanism with our running example, where $$\mathcal P_Q=\{\tuple{\emptyset,\underbrace{\{\m{Temp}(X,\m{high},T),\m{Temp}(X,\m{high},T+1),\m{Temp}(X,\m{high},T+2)\}}_H}\}\,.$$
  The atom with minimal timestamp in $H$ is $\m{Temp}(X,\m{high},T)$.
  We start by setting $\sch{-1}=\emptyset$.
  
  Since $D|_0=\{\m{Temp}(\m{wt25},\m{high},0)\}$, SLD-resolution between $\m{Temp}(X,\m{high},T)$
  and $D|_0$ yields $\subst{X:=\m{wt25},T:=0}$.
  Therefore,
  \[
  \sch0=
  \tuple{\subst{X:=\m{wt25},T:=0},\{\m{Temp}(\m{wt25},\m{high},0)\},
    \{\underbrace{\m{Temp}(\m{wt25},\m{high},1),\m{Temp}(\m{wt25},\m{high},2)}_{H_0}\}}\}\,.
  \]

  Next, $D|_1=\{\m{Temp}(\m{wt25},\m{high},1)\}$.
  As before, SLD-resolution between $\m{Temp}(X,\m{high},T)$ and $D|_1$ yields $\subst{X:=\m{wt25},T:=1}$.
  Furthermore, SLD-resolution between $\m{Temp}(\m{wt25},\m{high},1)$ and $D|_1$ yields $\emptyset$.

  Therefore,
  \begin{align*}
    \sch1=\{
    & \tuple{\subst{X:=\m{wt25},T:=1},\{\m{Temp}(\m{wt25},\m{high},1)\},
      \{\underbrace{\m{Temp}(\m{wt25},\m{high},2),\m{Temp}(\m{wt25},\m{high},3)}_{H'_1}\}}, \\
    & \tuple{\subst{X:=\m{wt25},T:=0},\{\m{Temp}(\m{wt25},\m{high},0),\m{Temp}(\m{wt25},\m{high},1)\},
      \{\underbrace{\m{Temp}(\m{wt25},\m{high},2)}_{H_1}\}}\}\,.
  \end{align*}

  Next we have $D|_2=\{\m{Temp}(\m{wt25},\m{high},2)\}$.
  Again, SLD-resolution between $\m{Temp}(X,\m{high},T)$ and $D|_2$ yields the substitution $\subst{X:=\m{wt25},T:=2}$, while SLD-resolution between $\m{Temp}(\m{wt25},\m{high},2)$ (the atom with minimal timestamp in both $H_1$ and $H'_1$) and $D_3$ yields $\emptyset$.
  Thus
  \begin{align*}
    \sch2=\{
    & \tuple{\subst{X:=\m{wt25},T:=2},\{\m{Temp}(\m{wt25},\m{high},2)\},
      \{\underbrace{\m{Temp}(\m{wt25},\m{high},3),\m{Temp}(\m{wt25},\m{high},4)}_{H''_2}\}}, \\
    & \tuple{\subst{X:=\m{wt25},T:=1},\{\m{Temp}(\m{wt25},\m{high},1),\m{Temp}(\m{wt25},\m{high},2)\},
      \{\underbrace{\m{Temp}(\m{wt25},\m{high},3)}_{H'_2}\}}, \\
    & \tuple{\subst{X:=\m{wt25},T:=0},\{\m{Temp}(\m{wt25},\m{high},0),\m{Temp}(\m{wt25},\m{high},1),\m{Temp}(\m{wt25},\m{high},2)\},
      \emptyset}\}\,.
  \end{align*}

  The last element of $\sch2$ has an empty set of hypotheses, so the answer $\subst{X:=\m{wt25},T:=0}$ is known at this point.

  Suppose now that $D|_3=\emptyset$.
  Then SLD-resolution between $\m{Temp}(X,\m{high},T)$ and $D|_3$ fails, so no new elements are added from $\mathcal P_Q$ to $\sch3$.
  Furthermore, $\m{Temp}(\m{wt25},\m{high},3)$ (which appears in both $H'_2$ and $H''_2$) does not unify with $D|_3$, and we know that it cannot be delivered by the data stream.
  Therefore
  \[ \sch3=\tuple{\subst{X:=\m{wt25},T:=0},\{\m{Temp}(\m{wt25},\m{high},0),\m{Temp}(\m{wt25},\m{high},1),\m{Temp}(\m{wt25},\m{high},2)\},\emptyset}\}\,.\]
  \eoe
\end{example}

This example illustrates that, in particular, answers are always propagated from $\sch\tau$ to $\sch{\tau+1}$.
In an actual implementation of this algorithm, we would likely expect these answers to be output when generated, and discarded afterwards.

\subsection{Correctness}
\label{sec:online-corr}

\begin{newtext}
We now show that the algorithm given in the previous section indeed maintains a set of schematic supported answers to $Q$ that captures all supported answers at any given time point.
\end{newtext}

\begin{theorem}[Soundness]
  \label{thm:sound}
  If $\tuple{\theta,E,H}\in\sch\tau$, $E\neq\emptyset$, and $\sigma$ instantiates all free variables in $E\cup H$, then $\tuple{\theta\sigma,H',E\sigma}\in\dsans$ for some $H'\subseteq H\sigma$.
\end{theorem}
\begin{proof}
  First, we show by induction on $\tau$ that there is an SLD-derivation with future premises of $Q$ and $D_\tau$ with computed answer with premises \tuple{\theta,H}.
  For $\sch{-1}=\emptyset$ this is trivially the case.

  Suppose now that $\tau\geq 0$ and $\tuple{\theta,E,H}\in\sch\tau$.
  By Lemma~\ref{lem:new-lemma}, \tuple{\theta,E,H} can be computed by SLD-resolution from $M\subseteq H'$ for either some $\tuple{\theta',H'}\in\mathcal P_Q$ or some $\tuple{\theta',E',H'}\in\sch{\tau-1}$.
  In both cases, we can compose this SLD-derivation with the one obtained either by the way $\mathcal P_Q$ is constructed or the one obtained by induction hypothesis to obtain an SLD-derivation with future premises of $Q$ and $D_\tau$ -- this simply requires adding the remaining elements of $H'$ to all nodes in the derivation from Lemma~\ref{lem:new-lemma}.
  
  By applying Theorem~\ref{thm:SLD-sound} to this SLD-derivation, we conclude that $\tuple{\theta\sigma,H'}\in\dhans$ for some $H'\subseteq H\sigma$.
  By construction, $E\sigma\neq\emptyset$ is evidence for this answer.
\end{proof}

It may be the case that $\sch\tau$ contains some elements that do not correspond to hypothetical answers because of the minimality requirement.
Consider a simple case of a query $Q$ where $\mathcal P_Q=\{\tuple{\emptyset,\{p(a,0),p(b,1),p(c,2)\}},\tuple{\emptyset,\{p(a,0),p(c,2),p(d,3)\}}\}$,
$D|_0=\{p(a,0)\}$ and $D|_1=\{p(b,1)\}$.
Then $\sch1=\{\tuple{\emptyset,\{p(a,0),p(b,1)\},\{p(c,2)\}},\tuple{\emptyset,\{p(a,0)\},\{p(c,2),p(d,3)\}}\}$, and the second element of this set has a non-minimal set of hypotheses.
We chose not to include a test for set inclusion in the definition of $\sch\tau$, though, for efficiency reasons.

\begin{theorem}[Completeness]
  \label{thm:compl}
  If $\tuple{\sigma,H,E}\in\dsans$, then there exist a substitution $\rho$ and a triple $\tuple{\theta,E',H'}\in\sch\tau$ such that $\sigma=\theta\rho$, $H=H'\rho$ and $E=E'\rho$.
\end{theorem}
\begin{proof}
  By Theorem~\ref{thm:SLD-compl}, \dcans\ for some substitution $\rho$ and set of atoms $H'=\{\alpha_i\}_i$ with $H=\{\alpha_i\rho\}_i$ and $\sigma=\theta\rho$ for some $\theta$.
  By Proposition~\ref{prop:SLD-strat}, there is a stratified SLD-derivation computing this answer.
  The sets of atoms from $D|_\tau$ unified in each stratum of this derivation define the set of elements from $H$ that need to be unified to construct the corresponding element of $\sch\tau$, and it is straightforward to check that they are defined as in Lemma~\ref{lem:new-lemma}.
\end{proof}

It also follows from our construction that: if $d$ is a query delay for $Q$, then the timestamp of each element of $H$ is at most $\tau+d$.
Likewise, if $w$ is a window size for $Q$, then all elements in $E$ must have timestamp at least $\tau-w$.

The following example also shows how, by outputting hypothetical answers, we can answer queries earlier than in other formalisms.

\begin{example}
  \label{ex:good}
  Suppose that we extend the program $\Piex$ in our running example with the following rule (as in Example~2 from~\cite{Motik}).
  \[
  \m{Temp}(X,\m{n/a},T) \to \m{Malf}(X,T)
  \]
  If $D_1=\{\m{Temp}(\m{wt25},\m{high},0),\m{Temp}(\m{wt25},\m{high},1),\m{Temp}(\m{wt42},\m{n/a},1)\}$, then
  \begin{align*}
    &\langle\subst{T:=0,X:=\m{wt25}},
      \{\m{Temp}(\m{wt25},\m{high},i)\mid i=0,1\},
      \{\m{Temp}(\m{wt25},\m{high},2)\}\rangle,\\
    &\langle\subst{T:=1,X:=\m{wt42}},\{\m{Temp}(\m{wt42},\m{n/a},1)\},\emptyset\rangle
  \end{align*}
  are both in $\sch1$.
  Thus, the answer \subst{T:=1,X:=\m{wt42}} is produced at timepoint~$1$, rather than being delayed until it is known whether \subst{T:=0,X:=\m{wt25}} is an answer.\eoe
\end{example}

\begin{theorem}[Complexity]
  \label{thm:algorithm}
  The set $\sch\tau$ can be computed from $\mathcal P_Q$ and $\sch{\tau-1}$ in time polynomial in the size of $\mathcal P_Q$, $\sch{\tau-1}$ and $D|_\tau$.
\end{theorem}
\begin{proof}
  If $\tuple{\theta,H}\in\mathcal P_Q$, then we can compute the set $M$ of elements of $H$ with minimal timestamp in linear time.
  To decide which substitutions make $M$ a subset of $D|_\tau$, we can perform classical SLD-resolution between $M$ and $D|_\tau$.
  For each such element of $\mathcal P_Q$, the size of every SLD-derivation that needs to be constructed is bound by the number of atoms in the initial goal, since $D|_\tau$ only contains facts.
  Furthermore, all unifiers can be constructed in time linear in the size of the formulas involved, since the only function symbol available is addition of temporal terms.
  Finally, the total number of SLD-derivations that needs to be considered is bound by the size of $\mathcal P_Q\times D|_\tau$.

  If $\tuple{\theta,E,H}\in\sch{\tau-1}$ and $E\neq\emptyset$, then again we can compute in linear time the set of facts in $H$ that must unify with $D|_\tau$ -- these are the elements of $H$ whose timestamp is exactly $\tau$.
  As above, the  elements that must be added to $\sch\tau$ can then be computed in polynomial time by SLD-resolution.
\end{proof}

It would be interesting to generalize our algorithm to non-connected queries, in order to deal with examples such as Example~\ref{ex:unbound}.
There, if $\m{Temp}(\m{wt25},\m{high},0)\in D_0$, it would be nice to obtain $$\tuple{\subst{X:=\m{wt25},T_1:=0},\{\m{Temp}(\m{wt25},\m{high},0)\},\{\m{Temp}(\m{wt25},\m{n/a},T_2)\}}\in\sch0$$ as this can be relevant information, even if we do not know when (or whether) it will lead to an answer to the original query.
However, such a generalization requires substantial changes to our approach, and we leave it for future work.

\section{Adding negation}
\label{sec:neg}

We now show how we can extend our framework to include negation.
We extend the syntax of our programs to allow negated atoms in the bodies of rules, but not in heads of rules or in queries.
The semantics of negation is based on the closed world assumption.

Hypothetical answers need to deal with the fact that this semantics introduces non-monotonicity: if $\neg p$ holds at a given point in time, $p$ may still become true later due to new incoming data.
Therefore $\neg p$ can only be taken as evidence for a hypothetical answer when the current $\tau$-history allows us to prove that $p$ does not hold in the whole dataset.
We deal with this issue by giving supported answers a semantics reminiscent of the Kripke semantics
for intuitionistic logic.

\begin{newtext}
\subsection{Temporal Datalog with negation}
To the best of our knowledge, the existing works on Temporal Datalog only consider the negation-free fragment of the language.
We extend the syntax of Temporal Datalog in the standard way~\cite{Abiteboul1995}, by allowing bodies of rules to contain negated atoms.
We further require any variables appearing in negated literals to occur elsewhere in a positive literal in the same rule (safe negation).

We illustrate this extended language with the example that we use throughout this section.

\begin{example}
  \label{ex:hospital}
  A hospital is conducting a study on early detection of critical patients.
  A patient that either shows bad results on their lab tests or does not have good vital signs (cardiac markers and blood oxygen levels) is placed under observation.
  A patient that is under observation and does not have good vital signs is moved to the ICU.
  This patient is considered a candidate for early risk detection two days before being moved to the ICU.

  These rules are implemented in the following program $\Pi_H$, which uses EDB predicates \m{BCA} (bad clinical analyses), \m{GCM} (good cardiac markers) and \m{GBOL} (good blood oxygen levels), as well as IDB predicates \m{GVS} (good vital signs), \m{ST} (patient status) and \m{Risk} (patient at risk).
  All these predicates include the patient's name as the first argument and the timestamp as the last; predicate \m{ST} also includes the patient status (\m{us} for ``under surveillance'' and \m{ic} for ``intensive care'') as second argument.
  Since the laboratory is overworked, they only communicate bad results of analyses as a combined outcome (predicate \m{BCA}) after two days.
  
  \begin{align*}
    \m{GCM}(X,T),\m{GBOL}(X,T) &\to \m{GVS}(X,T) \\
    \m{BCA}(X,T+2) &\to \m{ST}(X,\m{us},T+1) \\
    \neg\m{GVS}(X,T),\neg\m{ST}(X,\m{us},T) &\to \m{ST}(X,\m{us},T+1) \\
    \neg\m{GVS}(X,T),\m{ST}(X,\m{us},T) &\to \m{ST}(X,\m{ic},T+1) \\
    \m{ST}(X,\m{ic},T+2) &\to \m{Risk}(X,T)
  \end{align*}
  Note that $\m{BCA}(X,T+2)$ actually conveys information about the patient's status at timestamp $T$.
  (In particular, $\m{BCA}(X,0)$ and $\m{BCA}(X,1)$ are never produced by the datastream for any $X$.)
  \eoe
\end{example}

If $\Pi$ is a program and $D=\{D|_\tau\mid\tau\in\NN\}$ is a dataset, we write $\Pi\cup D\models\ell$ to denote that $\ell$ is entailed by $\Pi$ and $\bigcup_{\tau\in\NN}D|_\tau$ according to the cautious answer set semantics for logic programming~\cite{Gelfond1988}.
As usual, a precise correspondence between this semantics and an operational semantics based on SLD-resolution only holds in the presence of a stratification of the predicate symbols in the program with respect to negative dependencies~\cite{Lloyd1984}.
We now discuss what a notion of stratification should look like for programs in Temporal Datalog with negation.
\end{newtext}

\begin{definition}
  The \emph{temporal closure} of a program $\Pi$ in Temporal Datalog with negation is the program $\Pid$ defined as follows.
  For each $(n+1)$-ary predicate symbol $p$ with a temporal argument in the signature underlying $\Pi$, the signature for $\Pi^\downarrow$ contains a family of $n$-ary predicate symbols $\{p_t\}_{t\in\NN}$.
  For each rule in $\Pi$, $\Pid$ contains all rules obtained by instantiating its temporal parameter in all possible ways and replacing $p(x_1,\ldots,x_n,t)$ by $p_t(x_1,\ldots,x_n)$.
\end{definition}
Observe that $\Pid$ is an infinite program
\begin{newtext}%
(in Datalog with negation)    
\end{newtext}%
as long as $\Pi$ has at least one predicate with a temporal argument.

Recall that, for infinite programs, a stratification of $\Pi$ is an infinite sequence of disjoint programs $\Pi_0,\ldots,\Pi_n,\ldots$ such that $\Pi=\cup_{k\in\NN}\Pi_k$ and, for each predicate symbol $p$, (i)~the definition of $p$ is contained in one $\Pi_k$, (ii)~if $p$ occurs in the body of
a rule with head $q$, then the definition of $q$ is contained in $\Pi_i$ for some $i$ such that $i\geq k$ (if the occurrence is not negated) or $i>k$ (if the occurrence is negated)~\cite{Kolaitis1991}.

\begin{definition}
  A program $\Pi$ is $T$-stratified if $\Pid$ is stratified.
\end{definition}

This definition of stratification is reminiscent of the concepts of locally stratified~\cite{Przymusinski1988a,Przymusinski1988b} and temporally stratified~\cite{Zaniolo2015}.
However, it differs from both.
Local stratification has mostly been studied in the context of programs with a finite number of constants, which is not our case (there are infinitely many time points).
On the other hand, temporal stratification as in~\cite{Zaniolo2015} requires each stratum to contain exactly the predicates with the same time parameter (i.e., $\Pi_i$ contains all rules whose head uses predicate symbols indexed by $i$).
We make no such assumption in this work.

\begin{example}
  \label{ex:running-neg}
  Program $\Pi_H$ is $T$-stratified with the following strata:
  \begin{align*}
    \Pi_0=&\{\m{GCM}_t,\m{GBOL}_t,\m{GVS}_t,\m{BCA}_t\mid t\in\NN\}\cup\{\m{ST}_0\} \\
    \Pi_1=&\{\m{ST}_1\} \\
    \Pi_{t+2}=&\{\m{ST}_{t+2},\m{Risk}_t\} & t\geq 0
    \qquad\eoe
  \end{align*}
\end{example}

\begin{newtext}
Given an interpretation $I$, satisfaction of rules and programs is defined as before.
(Note that the restriction that negation be safe removes any potential ambiguity.)
It is easy to check that $I$ is a model of $\Pi$ iff it is a model of $\Pid$, and as a consequence $T$-stratified programs have a unique answer set that coincides with their well-founded model.
\end{newtext}

For finite programs, stratification can be phrased (and decided) as a property of the graph of dependencies between the predicate symbols in the program.
In our case, such a graph is infinite, and such a procedure may not terminate.
However, the usual construction can be adapted to our framework in order to provide a decision procedure.

\begin{newtext}
\begin{definition}
  Let $\Pi$ be a program.
  The \emph{temporal dependency graph} of $\Pi$, $\gr\Pi$, is the smallest graph satisfying the following properties:
  \begin{enumerate}
  \item for each predicate symbol $p$ with a temporal argument, $\gr\Pi$ includes a node with label $p(T)$ (where $T$ is a formal temporal variable);
  \item for each predicate symbol $p$ without temporal arguments, $\gr\Pi$ includes a node with label $p$;
  \item for each node with label $p(T+k)$ with $k\leq 0$ and each rule $r$ in the definition of $p$, let $\theta$ be the substitution that makes the temporal variable in the head of $r$ equal to $T+k$:
    \begin{enumerate}
    \item for each positive literal $q(x_1,\ldots,x_n,T+k')$ in $\body(r\theta)$, $\gr\Pi$ includes an edge between $p(T+k)$ and $q(T+k')$ with label $+$.
    \item for each negative literal $\neg q(x_1,\ldots,x_n,T+k')$ in $\body(r\theta)$, $\gr\Pi$ includes an edge between $p(T+k)$ and $q(T+k')$ with label $-$.
    \end{enumerate}
  \end{enumerate}
\end{definition}
The graph $\gr\Pi$ can be constructed by starting with the nodes described in the first two conditions, adding the edges (and new nodes) required by conditions 3(a) and 3(b), and iterating the process for any newly created nodes.
\end{newtext}

\begin{example}
  The temporal dependency graph for the program $\Pi_H$ is depicted below.
  \[
  \xymatrix@C-1em{
    &\m{ST}(T+2) \ar@/^/[d]^{-} \ar@/_/[d]_{+} \ar[drr]^{-}
    & \m{BCA}(T+2) \\
    & \m{ST}(T+1) \ar@/^/[d]^{-} \ar@/_/[d]_+ \ar[drr]^{-} \ar[ur]_+
    & \m{BCA}(T+1) & \m{GVS}(T+1) \ar@/^/[r]^+ \ar@/_1pc/[rr]_+
    & \m{GBOL}(T+1) & \m{GCM}(T+1) \\
    \m{Risk}(T) \ar@/^1pc/[uur]^+
    & \m{ST}(T) \ar@/^/[d]^{-} \ar@/_/[d]_+ \ar[drr]^{-} \ar[ur]_+
    & \m{BCA}(T) & \m{GVS}(T) \ar@/^/[r]^+ \ar@/_1pc/[rr]_+
    & \m{GBOL}(T) & \m{GCM}(T) \\
    & \m{ST}(T-1) \ar[ur]_+
    && \m{GVS}(T-1)
  }
  \]
  \eoe
\end{example}

\begin{proposition}
  \label{prop:stratification-dec}
  There is an algorithm that decides whether a program $\Pi$ is $T$-stratified.
\end{proposition}
\begin{proof}
  We first show that we can decide whether $\gr\Pi$ contains a path that passes infinitely many times through distinct edges labeled with $-$.
  Indeed, we can build a finite subgraph $\gr\Pi^-$ of $\gr\Pi$ by not expanding any nodes labeled $p(T+k)$ if $\gr\Pi^-$ already contains a path from $p(T+k')$ to $p(T+k)$ for some $k'<k$.
  Since $\Pi$ only contains a finite number of rules with a finite number of literals in their bodies, this condition generates an upper bound on the timestamps of the nodes that appear in the graph; since the number of predicate symbols is finite, $\gr\Pi^-$ must also be finite.
  Furthermore, $\gr\Pi$ contains a path in the conditions described if either (i)~$\gr\Pi^-$ also contains such a path (i.e.~it contains a loop including an edge labeled with $-$) or (ii)~$\gr\Pi^-$ contains a path from $p(T+k')$ to $p(T+k)$ with $k'<k$ that includes an edge labeled with $-$.
  
  Now we show that, if $\gr\Pi$ does not contain a path that passes infinitely many times through edges labeled with $-$, then $\Pi$ is $T$-stratified.

  Let $\gr\Pid$ be the graph of dependencies for $\Pid$.
  For each predicate symbol $p_t$ in $\Pid$, the graph $\gr\Pid$ contains a subgraph isomorphic to a subgraph\footnote{It may be a proper subgraph in case $t$ is low enough that some temporal arguments would become negative.} of the connected component of $\gr\Pi$ containing $p(T)$.
  In particular, $\gr\Pid$ contains infinitely many copies of $\gr\Pi$.

  Suppose that $\gr\Pid$ contains a path that passes infinitely many times through edges labeled with $-$, and let $p_t$ be a node in it with minimum value of $t$ (if $t$ occurs multiple times in the path, choose any one of its occurrences).
  Consider the isomorphic copy of $\gr\Pi$ embedded in $\gr\Pid$ where $p(T)$ is mapped to $p_t$.
  The whole path starting from $p_t$ must be contained in this copy, since the timestamps of all corresponding nodes in $\gr\Pi$ are greater or equal to $T$.

  Using these two properties, we can construct a stratification of $\Pid$ in the usual inductive way.
  \begin{itemize}
  \item $\mathcal G_0=\gr\Pid$
  \item For each $i\geq 0$, $\Pi_i$ contains all predicates labeling nodes $n$ in $\mathcal G_i$ such that: all paths from $n$ contain only edges labeled with $+$.
  \item For each $i\geq 0$, $\mathcal G_{i+1}$ is obtained from $\mathcal G_i$ by removing all nodes with labels in $\Pi_i$ and all edges to those nodes.
  \end{itemize}

  The only non-trivial part of showing that $\Pi_0,\ldots,\Pi_n,\ldots$ is a stratification of $\Pid$ is guaranteeing that every predicate symbol $p_t$ is contained in $\Pi_k$ for some $k$.
  This is established by induction: if any path from $p_t$ contains at most $n$ edges labeled with $-$, then $p_t\in\Pi_n$.
  We now show that, for any $p_t$, there is such a bound on the number of edges labeled with $-$.

  First observe that, in $\gr\Pi$, there is a bound $B$ on the number of edges labeled with $-$ in any path starting from $p(T)$: (i)~if there is a path from $q(T+k_1)$ to $q(T+k_2)$ for some $q$ and $k_1\leq k_2$ containing an edge labeled with $-$, then there would be a path with an infinite number of such edges, and (ii)~since there is a finite number of predicate symbols in $\Pi$ we cannot build paths with an arbitrarily long number of edges labeled with $-$.

  Now consider a path starting at $p_t$.
  If this path has length greater than $B$, then it cannot be contained completely in the subgraph of $\gr\Pid$ isomorphic to $\gr\Pi$ containing $p(T)$.
  Therefore, it must at some point reach a node $q_{t'}$ with $t'<t$.
  Likewise, the path starting at $q_{t'}$ corresponds to a path in a (different) subgraph of $\gr\Pid$ isomorphic to $\gr\Pi$, and as such cannot contain more than $B$ edges labeled with $-$.
  Since we can only decrease $t$ a finite number of times, the path cannot contain arbitrarily many edges labeled with $-$.
\end{proof}

Note that this does not contradict the classical undecidability of deciding whether a program has local stratification~\cite{Palopoli1992}, since our language does not have function symbols.

\begin{example}
  \begin{newtext}
  When the graph $\gr\Pi$ is finite, the proof of Proposition~\ref{prop:stratification-dec} can be used to construct any stratum $\Pi_n$ in finite time by constructing only the necessary part of $\gr\Pid$.
  In particular, applying it to $\gr{\Pi_H}$ shown earlier yields the stratification already presented in Example~\ref{ex:running-neg}.
  \end{newtext}
  
  For illustrative purposes, we show a fragment of the graphs $\mathcal G_0=\gr\Pid$, $\mathcal G_1$ and $\mathcal G_2$, marking in color the nodes that are removed (i.e., collected in $\Pi_i$).
  \[
  \mathcal G_0:
  \xymatrix{
    \m{Risk}_1 \ar[r]^+
    & \m{ST}_3 \ar@/^/[d]^{-} \ar@/_/[d]_+ \ar[drr]^{-}
    & \toremove{\m{BCA}_3}
    & \toremove{\m{GVS}_3} \ar@/^/[r]^+ \ar@/_1pc/[rr]_+
    & \toremove{\m{GBOL}_3}
    & \toremove{\m{GCM}_3} \\
    \m{Risk}_0 \ar[r]^+
    & \m{ST}_2 \ar@/^/[d]^{-} \ar@/_/[d]_+ \ar[drr]^{-} \ar[ur]_+
    & \toremove{\m{BCA}_2}
    & \toremove{\m{GVS}_2} \ar@/^/[r]^+ \ar@/_1pc/[rr]_+
    & \toremove{\m{GBOL}_2}
    & \toremove{\m{GCM}_2} \\
    & \m{ST}_1 \ar@/^/[d]^{-} \ar@/_/[d]_+ \ar[drr]^{-} \ar[ur]_+
    & \toremove{\m{BCA}_1}
    & \toremove{\m{GVS}_1} \ar@/^/[r]^+ \ar@/_1pc/[rr]_+
    & \toremove{\m{GBOL}_1}
    & \toremove{\m{GCM}_1} \\
    & \toremove{\m{ST}_0} \ar[ur]_+
    & \toremove{\m{BCA}_0}
    & \toremove{\m{GVS}_0} \ar@/^/[r]^+ \ar@/_1pc/[rr]_+
    & \toremove{\m{GBOL}_0}
    & \toremove{\m{GCM}_0}
  }
  \]

  \[
  \mathcal G_1:
  \xymatrix{
    \m{Risk}_1 \ar[r]^+
    & \m{ST}_3 \ar@/^/[d]^{-} \ar@/_/[d]_+ \\
    \m{Risk}_0 \ar[r]^+
    & \m{ST}_2 \ar@/^/[d]^{-} \ar@/_/[d]_+ \\
    & \toremove{\m{ST}_1}
  }
  \qquad
  \mathcal G_2:
  \xymatrix{
    \m{Risk}_1 \ar[r]^+
    & \m{ST}_3 \ar@/^/[d]^{-} \ar@/_/[d]_+ \\
    \toremove{\m{Risk}_0} \ar[r]^+
    & \toremove{\m{ST}_2} \\
  }
  \]
  \eoe
\end{example}

\subsection{Hypothetical answers in the presence of negation}

\begin{newtext}
The inclusion of negation in the language poses a new problem with respect to defining the notion of evidence for a hypothetical answer: we need to be able to argue that some atom will never become true given the set of hypotheses.
To formalize this notion, 
\end{newtext}
it is not enough to consider the data stream extended by $H$ -- we need to account also for additional information that may make more atoms provable.

\begin{definition}
  Let $D=\{D|_\tau\mid \tau\in\NN\}$ be a dataset, $\tau'$ be a time point and $H^+$ and $H^-$ be finite sets of ground EDB atoms.
  A dataset $D'=\{D'|_t\mid t\in\NN\}$ is a \emph{possible evolution} of $D$ compatible with \tuple{H^+,H^-} at time $\tau$ if:
  \begin{itemize}
  \item $D'|_\tau = D|_\tau$ for all $\tau\leq\tau'$;
  \item $H^+\subseteq\bigcup_{\tau\in\NN}D'|_t$;
  \item $H^-\cap\left(\bigcup_{\tau\in\NN}D'|_t\right)=\emptyset$.
  \end{itemize}
\end{definition}
Intuitively, $H^+$ contains the facts that \emph{must} be produced by the data stream, and $H^-$ the facts that \emph{cannot} be produced by the data stream.

Hypothetical answers to a query $Q=\tuple{\Pi,P}$ over $D_\tau$, where $\Pi$ may contain negated atoms, need to consider all possible evolutions of the dataset consistent with the hypotheses.

Given a set of literals $L$, we write $L^+$ for the set of positive literals in $L$ and $L^-$ for the set of atoms whose negation is in $L$.
In particular, we use this notation in the next two definitions.

\begin{definition}
  \label{defn:neg-hypothetical}
  A \emph{hypothetical answer} to query $Q=\tuple{P,\Pi}$ over $D_\tau$ is a pair $\tuple{\theta,H}$, where $\theta$ is a substitution and $H$ is a finite set of ground EDB literals such that:
  \begin{itemize}
  \item $\supp\theta=\var P$;
  \item $H^+$ and $H^-$ only contain atoms future-possible for $\tau$;
  \item $\Pi\cup\left(\bigcup_{\tau'\in\NN}D'|_{\tau'}\right)\models P\theta$ for each possible evolution $D'$ of $D$ compatible with \tuple{H^+,H^-} at time $\tau$;
  \item $H$ is minimal with respect to set inclusion.
  \end{itemize}
\end{definition}

\begin{definition}
  \label{defn:neg-supported}
  Let $Q=\tuple{P,\Pi}$ be a query, $D$ be a data stream and $\tau$ be a time instant.
  A set of ground EDB literals $E$ is \emph{evidence} supporting $\tuple{\theta,H}\in\dhans$ if:
  \begin{itemize}
  \item $E^+\subseteq D_\tau$ and $E^-\cap D_\tau=\emptyset$;
  \item $E^+\cup E^-\neq\emptyset$;
  \item $\Pi\bigcup D'\models P\theta$ for each possible evolution $D'$ of $D$ compatible with \tuple{E^+\cup H^+,E^-\cup H^-} at time~$-1$;
  \item $E$ is minimal with respect to set inclusion.
  \end{itemize}
  \begin{newtext}
  Note that the third condition effectively ignores $D$, since it ignores its contents at all time points.
  This matches our intuition that the evidence is exactly the  information from the data stream that is relevant to prove the query.
  \end{newtext}

  A \emph{supported answer} to $Q$ over $D_\tau$ is a triple \tuple{\theta,H,E}%
  \begin{newtext}%
  such that $\Pi\cup H\not\models P\theta$
  \end{newtext}%
  and $E$ is evidence supporting \tuple{\theta,H}.
\end{definition}
Intuitively, $E^+$ is the set of facts produced by the data stream that are essential to the hypothetical answer, while $E^-$ is the set of facts whose absence is relevant for the hypothetical answer.

We illustrate these alternative notions \begin{newtext}in the context of our\end{newtext} running example.
\begin{example}
  \label{ex:hospital-query}
  \begin{newtext}Consider the program $\Pi_H$ from Example~\ref{ex:hospital} and\end{newtext} the query $Q_R=\tuple{\m{Risk}(X,T),\Pi_H}$, and let $\theta=\subst{X:=\m{john},T:=0}$.
  \begin{itemize}
  \item Suppose that $D|_0=\emptyset$ then \tuple{\theta,H} where $H=\{\neg\m{GCM}(\m{john},1)\}$ is one (of several) hypothetical answers to $Q_R$ over $D_0$, with evidence $\{\neg\m{GBOL}(\m{john},0)\}$.

    Indeed $\m{GBOL}(\m{john},0)\notin D|_0$, so it can never hold; therefore $\neg\m{GVS}(\m{john},0)$ holds, and we can conclude $\m{ST}(\m{john},\m{us},1)$.
    If $\m{GCM}(\m{john},1)$ is never produced by the data stream, then $\neg\m{GVS}(\m{john},1)$ also holds, from which we can derive $\m{ST}(\m{john},\m{ic},2)$ and as a consequence $\m{Risk}(\m{john},0)$.

  \item Assume now that $D|_0=\{\m{GBOL}(\m{john},0),\m{GCM}(\m{john},0)\}$.
    Now there is a hypothetical answer \tuple{\theta,H'} to $Q_R$ over $D_1$, where $H'=\{\m{BCA}(\m{john},2),\neg\m{GCM}(\m{john},1)\}$.
    Indeed, $\m{ST}(\m{john},\m{us},1)$ holds as long as $\m{BCA}(\m{john},2)$ is produced by the data stream; and if $\m{GCM}(\m{john},1)$ is not produced by the data stream, then as before $\neg\m{GVS}(\m{john},1)$ holds, from which we can derive $\m{ST}(\m{john},\m{ic},2)$ and $\m{Risk}(\m{john},0)$.
    At this stage, there is no evidence for \tuple{\theta,H'}.
    \eoe
  \end{itemize}
\end{example}

\subsection{Pre-processing}

The notion of computed answer with premises in the presence of negation is generalized by allowing leaves of SLD-resolution to contain negated (EDB or IDB) atoms besides positive EDB atoms.

Our two-stage algorithm can now be adapted to the language with negation using this extended notion.
Let $Q=\tuple{P,\Pi}$ be a query.
In the pre-processing step, we compute the set $\mathcal P_Q$ using SLD-resolution as before.
For each element $\tuple{\theta,\{\alpha_i\}_{i\in I}}\in\mathcal P_Q$ and $i\in I$, if $\alpha_i$ is $\neg p(t_1,\ldots,t_n,t_{n+1})$, we generate a new query $Q'=\tuple{p(t_1,\ldots,t_n,T),\Pi}$, obtained by replacing the time parameter in $\alpha_i$ with a variable, and repeat the pre-processing step to compute $\mathcal P_{Q'}$.
We iterate this construction until no fresh queries are generated.

Given that the number of constants is finite and the number of arguments in any predicate in $\Pi$ is finite, the total number of queries that can be generated is also finite
\begin{newtext}-- albeit quite high, as in the worst case we can generate $O(m\times c^{k-1})$ auxiliary queries, where $m$ is the number of predicate symbols in the language, $c$ is the number of constants, and $k$ is the highest arity of any predicate in $\Pi$.\end{newtext}
(The only possible source of unboundedness is the temporal parameter, which we keep uninstantiated.)
Therefore, Propositions~\ref{prop:termination} and~\ref{thm:termination-char} still hold for this more expressive language,
\begin{newtext}but there is an exponential increase in worst-case complexity of the pre-processing step.\end{newtext}

\begin{example}
  \label{ex:preproc-neg}
  We continue with our running example (Example~\ref{ex:hospital}).
  Applying SLD-resolution to $\leftarrow\m{Risk}(X,T)$ until we reach a goal with only extensional or negated atoms results in two possible derivations:
  \[\xymatrix@R=1em@C=-6em{
    & \leftarrow\m{Risk}(X,T) \ar[d] \\
    & \leftarrow\m{ST}(X,\m{ic},T+2) \ar[d] \\
    & \leftarrow\neg\m{GVS}(X,T+1),\m{ST}(X,\m{us},T+1) \ar[]+D;[dl] \ar[]+D;[dr] \\
    \leftarrow\neg\m{GVS}(X,T+1),\m{BCA}(X,T+2)
    && \leftarrow\neg\m{GVS}(X,T+1),\neg\m{GVS}(X,T),\neg\m{ST}(X,\m{us},T)
  }\]
  This yields the following set $\mathcal P_{Q_R}$.
  \begin{align*}
    \mathcal P_{Q_R} &= \{
      \tuple{\emptyset,\{\m{BCA}(X,T+2),\neg\m{GVS}(X,T+1)\}},
      \\ & \phantom{{}=\{}
      \tuple{\emptyset,\{\neg\m{GVS}(X,T+1),\neg\m{GVS}(X,T),\neg\m{ST}(X,\m{us},T)\}}.
  \end{align*} 

  Since \m{GVS} and \m{ST} appear as negated hypotheses, we generate two auxiliary queries $Q_G=\tuple{\m{GVS}(X,T),\Pi_H}$ and $Q_S=\tuple{\m{ST}(X,\m{us},T),\Pi_H}$.
  These queries are now pre-processed in turn.
  \begin{align*}
    \mathcal P_{Q_G} &= \tuple{\emptyset,\{\m{GCM}(X,T),\m{GBOL}(X,T)\}} \\
    \mathcal P_{Q_S} &= \{\tuple{\emptyset,\{\m{BCA}(X,T+1)\}},
      \\ & \phantom{{}=\{}
      \tuple{\emptyset,\{\neg\m{GVS}(X,T-1),\neg\m{ST}(X,\m{us},T-1)\}}
  \end{align*}
  Since the newly generated negated goals yield queries that have already been considered, the pre-processing step is finished.
  \eoe
\end{example}

\paragraph{Constants in generated queries}
The negated atoms in leaves may be partly instantiated.
This creates a need to decide whether those instantiations should be maintained, or whether all terms should be replaced by variables in the new query.
The precise choice depends on the concrete instance of the problem: replacing all terms by variables restricts the total number of queries being managed, since all negated instances involving the same predicate reduce to the same query.
However, having instantiated terms may substantially reduce the number of branches in the SLD-tree constructed, and the number of potential answers that have to be tracked.
In particular, if the instance originates from a term that is instantiated in the original query, it might be counter-productive to replace it by a variable.

\begin{example}
  Consider the program $\Pi$ below and the query \tuple{\m{Q}(T),\Pi}.
  \begin{align*}
    \neg\m{P}(\m{c_1},T) &\to \m{Q}(T) &
    \neg\m{P}(\m{c_2},T) &\to \m{Q}(T) \\
    \neg\m{P}(\m{c_3},T) &\to \m{Q}(T) &
    \m{R}(X,T) &\to \m{P}(X,T)
  \end{align*}

  In this case, replacing the constants \m{c_1}, \m{c_2} and \m{c_3} with a variable in the generated queries results in the same query $\m{P}(X,T)$ in all cases.
  Keeping the constants in the generated queries results in the need to process $\m{P}(\m{c_1},T)$, $\m{P}(\m{c_2},T)$ and $\m{P}(\m{c_3},T)$, which generates three isomorphic trees and nearly identical computed answers with premises.\eoe
\end{example}

\begin{example}
  Consider the following program $\Pi$
  \begin{align*}
    \neg\m{P}(X,T) &\to \m{Q}(X,T)
    &&\to \m{P}(\m{c},T) \\
    && \m{R}(X,T) &\to \m{P}(X,T)
  \end{align*}
  and the query \tuple{\m{Q}(\m{c},T),\Pi}.
  This query generates as premise $\neg\m{P}(\m{c},T)$, which can never hold because of the second rule.
  However, if $\m{c}$ is abstracted to a variable, the generated query becomes $\m{P}(X,T)$, which yields two computed answer with premises $\tuple{\subst{X:=\m{c}},\emptyset,\emptyset}$ and $\tuple{\emptyset,\m{R}(X,T),\emptyset}$.
  \eoe
\end{example}

One possible approach would be to keep all instantiated non-temporal terms unchanged, but replace them by variables if a new query on the same predicate is generated.
We defer the concrete choice of strategy to future work, as it is immaterial for this work.
In particular, the previous comment on Propositions~\ref{prop:termination} and~\ref{thm:termination-char} applies in either case.

\begin{theorem}[Soundness and completeness of pre-processing]
  \label{thm:preproc-sound-compl}
  Let $Q^\ast=\tuple{\Pi,P}$ be one of the queries obtained from the pre-processing of $Q$, $F$ be a set of literals with timestamp higher than $\tau$ such that $F^+$ only contains EDB atoms.
  Assume that $\Pi\cup D_\tau\cup F$ is consistent, and let $\theta$ be a substitution that instantiates all variables in $P$.
  Then $\Pi\cup D_\tau\cup F\models P\theta$ iff there exist substitutions $\theta'$ and $\sigma$ and a set $H$ such that $\tuple{\theta',H}\in\mathcal P_{Q^\ast}$, $\theta=\theta'\sigma$ and $H\sigma\subseteq F$.
\end{theorem}
\begin{proof}
  From soundness and completeness of SLD-resolution, it immediately follows that $\Pi\cup D_\tau\cup F\models P\theta$ iff $\tuple{\theta',H}\in\mathcal P_Q$ for some $\theta'$ and $H$ such that $\theta=\theta'\sigma$ and $H\sigma\subseteq F$ for some substitution $\sigma$.
  (Note that the pre-processing step essentially treats negated literals as EDB atoms.)
\end{proof}

\subsection{The iterative step}
\label{sec:neg-iter}

In the online step of the algorithm, we now need to compute schematic answers not only to the original query, but also to each query generated by the pre-processing step.
In this section, we use $Q^\ast$ to range over all these queries.
This step is now significantly more complex than for the language without negation: besides updating the (positive) hypotheses in hypothetical answers with information from the data stream, we need to apply a fixpoint construction in order to update the negative hypotheses -- these are updated using the computed schematic answers themselves, and each update can trigger new possible updates.

In order to simplify the presentation, we first formulate the iterative step working only with ground terms.
After discussing its soundness and completeness, we show how it can be reformulated in terms of schematic hypothetical answers using variables, as in the previous case.
Without loss of generality, we also assume that all constants are abstracted to variables in the queries generated by the pre-processing step, so that for each literal $\ell$ there exists at most one query on the predicate symbol in $\ell$.
(This assumption removes the need for some quantifications, but does not change anything in an essential manner.)

We work with \emph{generalized hypothetical answers}, where hypotheses can be EDB literals or negated IDB atoms, and we recursively define sets $\mathcal S^\downarrow_\tau(Q^\ast)$ of these entities.
The definition uses an auxiliary family $\mathcal A_\tau(Q^\ast)$; intuitively, $\mathcal A_\tau(Q^\ast)$ is computed from $\mathcal S^\downarrow_{\tau-1}(Q^\ast)$ by updating positive hypotheses according to the data stream (as in the case without negation), and $\mathcal S^\downarrow_\tau(Q^\ast)$ is computed from $\mathcal A_\tau(Q^\ast)$ by updating the
negative hypotheses.

\begin{definition}
  \label{defn:Sdownarrow}
  The families of sets $\mathcal S^\downarrow_\tau(Q^\ast)$ and $\mathcal A_\tau(Q^\ast)$ are defined recursively as follows, where $Q^\ast=\tuple{P,\Pi}$ is one of the queries under consideration.
  \begin{enumerate}
  \item $\mathcal S^\downarrow_{-1}(Q^\ast)=\emptyset$.
  \item For each $\tuple{\theta,E,H}\in\mathcal S^\downarrow_{\tau-1}{Q^\ast}$, let $M$ be the (possibly empty) set of elements of $H^+$ with timestamp $\tau$.
    If $M\subseteq D|_\tau$, $\tuple{\theta,E\cup M,H\setminus M}\in\mathcal A_\tau(Q^\ast)$.
  \item For each $\tuple{\theta,H}\in\mathcal P_{Q^\ast}$, let $M$ be the set of elements of $H$ with minimal timestamp.
    For each substitution $\sigma$ such that $H\sigma$ is ground:
    \begin{enumerate}
    \item If $M^+$ is non-empty and $M^+\sigma\subseteq D|_\tau$, then $\tuple{\theta\sigma,M^+\sigma,(H\setminus M^+)\sigma}\in\mathcal A_\tau(Q^\ast)$.
    \item If $P\sigma$ has timestamp $\tau$ and all elements of $M\sigma$ are negative literals, then $\tuple{\theta\sigma,\emptyset,H\sigma}\in\mathcal A_\tau(Q^\ast)$.
    \item If $P\sigma$ has timestamp $\tau$ and every element of $H\sigma$ has timestamp strictly larger than $\tau$, then $\tuple{\theta\sigma,\emptyset,H\sigma}\in\mathcal A_\tau(Q^\ast)$.
    \end{enumerate}
  \item $\mathcal S^\downarrow_\tau(Q^\ast)$ is obtained by applying the following operator $\mathcal R$ to $\mathcal A_\tau(Q^\ast)$ until a fixpoint is reached: given a family $X(Q^\ast)$ of generalized hypothetical answers, non-deterministically select a query $Q_\ell$ over a literal $\ell$.
    \begin{enumerate}
    \item If there is a tuple $\tuple{\theta,E_\ell,\emptyset}\in X(Q_\ell)$, then for all $Q^\ast$ remove every \tuple{\sigma,E,H} from $X(Q^\ast)$ whenever $\ell\theta\in H^-$.
    \item For all substitutions $\sigma$ such that the timestamp of $\ell\sigma$ is at most $\tau$ and there is no tuple $\tuple{\sigma,E_\ell,H_\ell}\in X(Q_\ell)$, if $X(Q^\ast)$ contains an element \tuple{\theta,E,H} such that $\ell\sigma\in H^-$, then remove $\neg\ell\sigma$ from $H$ and add it to $E$.
    \end{enumerate}
    We assume that $\mathcal R$ selects a query $Q_\ell$ such that $X$ changes, if possible.
  \end{enumerate}
\end{definition}

The operator $\mathcal R$ randomly selects a query $Q_\ell$ and updates the generalized hypothetical answers that use negative ground instances of literal $\ell$ by (i)~discarding the ones whose hypotheses contradict generated answers to $Q_\ell$ and (ii)~moving hypotheses to evidence when they match no generalized hypothetical answer to $Q_\ell$.
In order to guarantee that all these steps are finite, the second step in the computation of $\mathcal R(X)$ only looks at hypotheses whose timestamp is at most the current value of $\tau$ -- steps~3(b) and~3(c) ensure that $\mathcal R$ does not incorrectly remove some negative hypotheses.

\begin{lemma}
  Iterating $\mathcal R$ always reaches a fixpoint in finite time.
\end{lemma}
\begin{proof}
  If $\mathcal R(X)$ is different from $X$, then either the total number of elements in all the sets $X(Q^\ast)$ decreases or the total number of elements in all the sets $H$ in tuples \tuple{\theta,E,H} in all $X(Q^\ast)$ decreases.
  But set $X$ is finite: every element of $\mathcal A_\tau(Q^\ast)$ either comes from unification between $\mathcal P_{Q^\ast}$ and the current $\tau$-history (and since both $\mathcal P_{Q^\ast}$ and $\tau$-histories are finite, only a finite number of elements can be generated in this way), or from directly instantiating elements of $\mathcal P_{Q^\ast}$ with a timestamp at most $\tau$ (which can only be done in finitely many ways).
  Therefore this process must terminate.
\end{proof}

\begin{lemma}
  \label{lem:R-unique}
  For each query $Q^\ast$, the set $\mathcal S^\downarrow_\tau(Q^\ast)$ is uniquely defined.
\end{lemma}
\begin{proof}
  We start by showing that: if we mark the generalized hypothetical answers that can be removed from a set $X$ and the set of (negative) hypotheses that can be turned into evidence, then, after applying $\mathcal R$ once, all actions that were not made can still be applied afterwards.
  Together with the fact that the actions enabled by new tuples that arrive are the same, this implies that all actions must have been performed when reaching a fixpoint.
  
  Let $X$ and $Y$ be families of generalized hypothetical answers to the queries $Q^\ast$ such that $Y$ can be obtained from $X$ by application of $R$.
  \begin{itemize}
  \item If \tuple{\sigma,E,H} can be removed from $X(Q^\ast)$ by application of $\mathcal R$ and there is a corresponding tuple \tuple{\sigma,E',H'} in $Y$ (either the same or obtained by moving something from $H$ to $E$ in the second case defining $\mathcal R$), then \tuple{\sigma,E',H'} can be removed from $Y(Q^\ast)$ by application of $\mathcal R$.

    Indeed, if $\tuple{\theta,E_\ell,\emptyset}\in X(Q_\ell)$, then this answer must still be in $Y(Q_\ell)$, since $\mathcal R$ does not remove answers (with empty set of hypotheses).
    Furthermore, if $\ell\theta\in H^-$ then $\ell\theta\in {H'}^-$, since \tuple{\theta,E_\ell,\emptyset} prevents $\ell\theta$ from being removed from $H$ by application of $\mathcal R$.

  \item If $\tuple{\sigma,E,H}\in X(Q^\ast)$ can be transformed into \tuple{\sigma,E\cup\{\ell\theta\},H\setminus\{\ell\theta\}} by application of $\mathcal R$ and there is a corresponding tuple \tuple{\sigma,E',H'} in $Y$ with $\ell\theta\in H'$, then \tuple{\sigma,E',H'} can be changed to \tuple{\sigma,E'\cup\{\ell\theta\},H'\setminus\{\ell\theta\}} by application of $\mathcal R$.

    Indeed, if the timestamp of $\ell\theta$ is at most $\tau$ and there is no tuple
    $\tuple{\theta,E,H}\in X(Q_\ell)$, then this must still be the case in $Y(Q_\ell)$, since
    $\mathcal R$ does not add new tuples to $X$.
  \end{itemize}

  Now consider two sequences of applications of $\mathcal R$ starting from $X$ and ending at a fixpoint.
  By inductively applying the previous argument, any tuples that can be removed or changed in $X$ must have been removed or changed in the same way when reaching the fixpoint.
  Furthermore, new tuples that can be removed or changed due to tuples that were removed or changed will be the same in both sequences of applications, since these tuples are directly determined by the action(s) performed.
  Therefore both sequences must end at the same fixpoint.
\end{proof}

We illustrate this construction by means of a small example, before moving to our richer running example.

\begin{example}
  \label{ex:negation-past}
  \begin{newtext}Consider the query $Q=\tuple{\m{R}(T),\Pi'}$ where $\Pi'$ is the following program.\end{newtext}
  \begin{align*}
    \m{S}(T), \m{S}(T+1) &\to \m{P}(T) &
    \neg\m{P}(T) &\to \m{R}(T)
  \end{align*}

  The pre-processing step for this \begin{newtext}query\end{newtext} generates $\mathcal P_Q=\{\tuple{\emptyset,\{\neg\m{P}(T)\}}\}$.
  Since $\m{P}(T)$ occurs negated in the premise of the only answer in this set, we generate a new query $Q'=\tuple{\m{P}(T),\Pi}$, for which the pre-processing step yields $\mathcal P_{Q'}=\{\tuple{\emptyset,\{\m{S}(T),\m{S}(T+1)\}}\}$.
  
  Assume that $D_0=\emptyset$.
  Since the only element of $P_Q$ contains a negated atom in its premises, we are in case (3c) and $\mathcal A_0(Q)=\tuple{\subst{T:=0},\emptyset,\{\neg\m{P}(0)\}}$.
  Furthermore, $\mathcal A_0(Q')=\emptyset$, since $\m{S}(0)\not\in D_0$.
  Since there is no element in $\mathcal A_0(Q')$, we move $\neg\m{P}(0)$ to the set of evidence in the only element of $\mathcal A_0(Q)$; this yields a fixpoint of $\mathcal R$, so $\sch0(Q)=\tuple{\subst{T:=0},\{\neg\m{P}(0)\},\emptyset}$ and $\sch0(Q')=\emptyset$.

  Now suppose that $D_1=\{\m{S}(1)\}$.
  Reasoning as before, we obtain $\mathcal A_1(Q)=\{\tuple{\subst{T:=1},\emptyset,\{\neg\m{P}(1)\}}\}$.
  However, since $\m{S}(1)\in D_1$, we also obtain $\mathcal A_1(Q')=\{\tuple{\subst{T:=1},\{\m{S}(1)\},\{\m{S}(2)\}}\}$.
  Now there is an answer in $\mathcal A_1(Q')$ with substitution $\subst{T:=1}$ and non-empty set of premises, so we cannot conclude anything about $\m{P}(1)$ yet, and thus $\mathcal A_1$ is already a fixpoint of $\mathcal R$; so $\sch1=\mathcal A_1$.

  If $D_2=\{\m{S}(2)\}$, then $\tuple{\subst{T:=1},\{\m{S}(1),\m{S}(2)\},\emptyset}\in\mathcal A_2(Q')$.
  As a consequence, $\neg\m{P}(1)$ does not hold, and $\tuple{\subst{T:=1},\emptyset,\{\neg\m{P}(1)\}}\in\mathcal A_2(Q)$ is removed by $\mathcal R$.

  On the other hand, if $D_2=\emptyset$, it is the answer $\tuple{\subst{T:=1},\{\m{S}(1)\},\{\m{S}(2)\}}\in\sch1(Q')$ that does not propagate to $\mathcal A_2(Q')$, from which applying $\mathcal R$ will allow us to conclude that $\subst{T:=1}$ is output as an answer to $Q$.
  \eoe
\end{example}

Soundness and completeness hold for $T$-stratified programs.
Proving soundness requires a partial form of completeness, which is why the next lemma is stated as an if-and-only-if.

\begin{theorem}[Soundness]
  \label{thm:negated-sound}
  Let $Q^\ast=\tuple{\Pi,P}$, where $\Pi$ is $T$-stratified, be one of the queries obtained from the pre-processing of $Q$, $F$ be a set of literals with timestamp higher than $\tau$ such that $F^+$ only contains EDB atoms.
  Assume that $\Pi\cup D_\tau\cup F$ is consistent, and let $\theta$ be a substitution that instantiates all variables in $P$.

  Then:
  \begin{enumerate}[(i)]
  \item If $\tuple{\theta,E,H}\in\mathcal S^\downarrow_\tau(Q^\ast)$ for some $E$ and $H$ such that $E^+\subseteq D_\tau$ and $H\subseteq F$, then $\Pi\cup D_\tau\cup F\models P\theta$.
  \item If the timestamp of $P\theta$ is at most $\tau$, then the converse implication also holds.
  \end{enumerate}
\end{theorem}
\begin{proof}
  The proof is by induction on the construction of $\mathcal S^\downarrow_\tau(Q^\ast)$, considering all queries $Q^\ast$ simultaneously.
  For clarity, we prove (i) and (ii) separately arguing by induction on the timestamp, but their proofs use each other's induction hypothesis.
  We point out that induction hypotheses are always applied to sets that have been constructed earlier, i.e.~we first conclude (i) and (ii) for $\mathcal S^\downarrow_{-1}$, then for $\mathcal A_0$, then for $\mathcal R(\mathcal A_0)$, then for $\mathcal R(\mathcal R(\mathcal A_0))$, etc, until we reach $\mathcal S^\downarrow_0$.

  We first prove property~(i).
  For $\tau=-1$, this is trivial.
  Assume that the result holds for $\sch{\tau-1}$.
  We go through the steps of the construction of $\sch\tau(Q^\ast)$.
  \begin{enumerate}[(Step 1.)]
    \addtocounter{enumi}1
  \item Suppose that $\tuple{\theta,E,H}\in\mathcal S^\downarrow_{\tau-1}(Q^\ast)$, that $M\subseteq D|_\tau$ for the set $M$ of elements with timestamp $\tau$ of $H$, and that $E^+\cup M\subseteq D_\tau$ and $H\setminus M\subseteq F$.
    By construction $E^+\subseteq D_{\tau-1}$, and $H\subseteq F\cup M\subseteq F\cup D|_\tau$.
    If $\Pi\cup D_\tau\cup F$ is consistent, then the induction hypothesis yields the thesis, since $\Pi\cup D_\tau\cup F=\Pi\cup D_{\tau-1}\cup(F\cup D|_\tau)$.
  \item Suppose that $\tuple{\theta,H}\in\mathcal P_{Q^\ast}$, let $M$ be the set of elements in $H$ with minimal timestamp, and suppose that $\sigma$ makes all elements of $H\sigma$ ground.
    The cases~(b) and~(c) follow immediately from Theorem~\ref{thm:preproc-sound-compl}.

    For case~(a), suppose that $M^+\sigma\subseteq D|_\tau$ and that $(H\setminus M^+)\sigma\subseteq F$.
    In particular, all elements of $M$ have timestamp $\tau$.
    Then $H\sigma\subseteq F\cup M^+\sigma\subseteq F\cup D|_\tau$, so by Theorem~\ref{thm:preproc-sound-compl}, $\Pi\cup D_{\tau-1}\cup(F\cup D|_\tau)\models P\theta\sigma$.
    The thesis follows by observing that $\Pi\cup D_{\tau-1}\cup(F\cup D|_\tau)=\Pi\cup D_\tau \cup F$.
  \item The proof follows by induction on the number of times that $\mathcal R$ is applied.
    Again, step~(a) is trivial, since it generates no new generalized hypothetical answers.

    For step~(b), we need to look at the generalized hypothetical answers that are changed by $\mathcal R$.
    Assume that $\tuple{\theta,E,H}$ is such that $\ell\sigma\in H^-$ for some $\ell$ and $\sigma$ with the timestamp of $\ell\sigma$ at most $\tau$, and that there is no tuple $\tuple{\sigma,E_\ell,H_\ell}\in X(Q_\ell)$.
    Suppose that $(E\cup\{\neg\ell\sigma\})^+\subseteq D_\tau$ and $H\setminus{\neg\ell\sigma}\subseteq F$.
    Then $E^+\subseteq D_\tau$ and $H\subseteq F\cup\{\neg\ell\sigma\}$.

    But $\Pi\cup D_\tau\cup(F\cup\{\neg\ell\sigma\})$ is consistent iff $\Pi\cup D_\tau\cup F\not\models\ell\sigma$, which by induction hypothesis (on $Q_\ell$) is the case since there are no tuples $\tuple{\sigma,E_\ell,H_\ell}\in X(Q_\ell)$.
    Therefore the induction hypothesis (on $Q^\ast$) allows us to conclude that $\Pi\cup D_\tau\cup(F\cup\{\neg\ell\sigma\})\models P\theta$.
    But the induction hypothesis (on $Q_\ell$) also yields that $\Pi\cup D_\tau\cup F\models\neg\ell\sigma$: since $\Pi$ is stratified, if this were not the case then $\Pi\cup D_\tau\cup F\models\ell\sigma$, whence there would be some tuple $\tuple{\sigma,E_\ell,H_\ell}\in X(Q_\ell)$.
  \end{enumerate}

  We now move to property~(ii).
  From Theorem~\ref{thm:preproc-sound-compl}, there exists $\tuple{\theta',H_0}\in\mathcal P_{Q^\ast}$ such that $\theta=\theta'\sigma$ for some substitution $\sigma$ and $H_0\sigma\subseteq F$.
  We prove that~(ii) holds whenever the minimal timestamp of $H_0^+\sigma\cup\{P\theta\}$ is at most $\tau$, which implies the thesis.
  For $\tau=-1$ this vacuously holds, since timestamps must be positive.

  Assume that $\Pi\cup D_\tau\cup F\models P\theta$.
  If the minimal timestamp of $H_0^+\sigma\cup\{P\theta\}$ is strictly smaller than $\tau$, then by induction hypothesis applied to $F\cup D|_\tau$, there exists $\tuple{\theta,E,H}\in\mathcal S^\downarrow_{\tau-1}(Q^\ast)$ such that $E^+\subseteq D_{\tau-1}$ and $H\subseteq F\cup D|_\tau$.
  Step~(2) of the algorithm then adds $\tuple{\theta,E\cup M,H\subseteq M}$ to $\mathcal A_\tau(Q^\ast)$, where $M\subseteq D|_\tau$, and by construction $(E\cup M)^+\subseteq D_\tau$ and $H\subseteq F$.
  If the minimal timestamp of $H_0^+\sigma\cup\{P\theta\}$ is exactly $\tau$, then a tuple satisfying the thesis is added to $\mathcal A_\tau(Q^\ast)$ in step~(3).

  We now show that the fixpoint construction in step~(4) does not break this property.
  \begin{itemize}
  \item Suppose step~4(a) applies to a set $X$.
    If there is a tuple $\tuple{\theta,E_\ell,\emptyset}\in X(Q_\ell)$, then by induction hypothesis $\Pi\cup D_\tau\cup F'\models\ell\theta$ for any set $F'$ in the conditions of the theorem, and therefore $\Pi\cup D_\tau\cup F$ is inconsistent whenever $\neg\ell\theta\in F$.
    As such, any tuple whose set $H$ contains $\neg\ell\theta$ can be removed from $X$.
  \item Suppose step~4(b) applies to a set $X$.
    Since this step adds elements to $E^-$ and removes elements from $H$, the thesis holds even if this step is applied to this tuple.\qedhere
  \end{itemize}
\end{proof}

This result guarantees that any answer to query $Q$ will eventually be represented by an element of $\mathcal S^\downarrow_\tau(Q)$, but it does not ensure that that element necessarily becomes an answer, since the set of premises may never become empty.
The next result shows that this is indeed the case.

\begin{proposition}
  \label{lem:negated-compl}
  In the same conditions as the previous theorem, if $\tuple{\theta,E,H}\in\mathcal S^\downarrow_\tau(Q^\ast)$, then there exists $\tau'\geq\tau$ such that either (a)~$\tuple{\theta,E\cup H,\emptyset}\in\mathcal S^\downarrow_{\tau'}(Q^\ast)$ or (b)~$\mathcal S^\downarrow_{\tau'}(Q^\ast)$ does not contain any element of the form \tuple{\theta,E',H'} such that $E'\cup H'=E\cup H$.
\end{proposition}
\begin{proof}
  By induction on the stratification of $\Pi$, i.e., on the stratum of $P_t$, where $t$ is the (instantiated) temporal argument in $P\theta$.

  If $P_t$ is in the lowest stratum, then it does not depend on any negated predicates.
  The thesis then follows by observing that the algorithm behaves for this predicate as in the case of the fragment without negation, and invoking Theorems~\ref{thm:sound} and~\ref{thm:compl}.

  Otherwise, by induction hypothesis, for every negative literal in $H$ there is a timestamp where the corresponding query is decided (any schematic answer for the relevant predicate has no premises).
  Take $\tau_0$ to be the highest of those timestamps.
  In step~(4) of the algorithm, in any generalized hypothetical answers where $\theta$ instantiates the timestamp of $P$ to $t$, all negated literals in the set of premises must either lead to the answer being discarded 4(a), or be moved to the set of evidence 4(b).
  Therefore, at the end of this iteration, for any generalized hypothetical answer $\tuple{\theta,E',H'}\in\mathcal S^\downarrow_{\tau_0+1}(Q^\ast)$ the set $H'$ only has positive literals.
  If $H'$ is always empty, the thesis holds.
  Otherwise, let $\tau'$ be the highest timestamp in all such sets $H'$; the construction in step~(3) guarantees that in $\mathcal S^\downarrow_{\tau'}(Q^\ast)$ all $H'$ in answers that have not been discarded must be empty.
\end{proof}

\begin{corollary}[Completeness]
  If $\theta$ is an answer to $Q$, then there exists $\tau$ such that $\sch{\tau}(Q)$ contains a tuple \tuple{\theta',E,\emptyset} where $\theta'=\theta\sigma$ for some substitution $\sigma$.
\end{corollary}

\paragraph{Remark.}
The correspondence with the declarative notion of (supported) hypothetical answer is less direct, since generalized hypothetical answers are allowed to have negated literals.
Variants of Theorem~\ref{thm:negated-sound} and Lemma~\ref{lem:negated-compl} relating directly to hypothetical answers can be obtained by recursively processing negative hypotheses and evidence:
\begin{itemize}
\item if $\neg\ell$ appears in a set of evidence, replace it with the evidence for $\ell$ that was used to place it there in step~4(b);
\item if $\neg\ell\theta$ appears in a set of hypotheses, then replace the generalized hypothetical answer with all possible generalized hypothetical answers where $\neg\ell$ is replaced by $\neg\ell'$ with $\tuple{\theta,E,H}\in\sch\tau(Q^\ell)$ and $\ell'\in H$.
\end{itemize}
The last step generates a combinatorial explosion of the number of hypothetical answers.
We also believe that generalized hypothetical answers are more readable in practice, as they encapsulate the strata of the program (and the information about negative hypotheses/evidence can be inspected separately).
For these reasons, we do not pursue this correspondence further.

We now present an alternative to Definition~\ref{defn:Sdownarrow} avoiding grounding and discuss its complexity.
\begin{definition}
  \label{defn:Salt}
  The sets $\sch\tau(Q^\ast)$ and $\mathcal B_\tau(Q^\ast)$, where $Q^\ast=\tuple{P,\Pi}$ is one of the queries under consideration, are recursively defined as follows.
  \begin{enumerate}
  \item $\sch{-1}(Q^\ast)=\emptyset$.
  \item For each $\tuple{\theta,E,H}\in\sch{\tau-1}(Q^\ast)$, let $M$ be the (possibly empty) set of elements of $H^+$ with timestamp $\tau$.
    For all substitutions $\sigma$ such that $M\sigma\subseteq D|_\tau$, $\tuple{\theta,E\cup M\sigma,(H\setminus M)\sigma}\in\mathcal B_\tau(Q^\ast)$.
  \item For each $\tuple{\theta,H}\in\mathcal P_{Q^\ast}$, let $M$ be the set of elements of $H$ with minimal timestamp.
    \begin{enumerate}
    \item If $M^+$ is non-empty, then for every substitution $\sigma$ such that $M^+\sigma\subseteq D|_\tau$, then $\tuple{\theta\sigma,M^+\sigma,(H\setminus M^+)\sigma}\in\mathcal B_\tau(Q^\ast)$.
    \item Let $\sigma=\subst{T:=\tau}$ with $T$ the temporal variable in $P$.
      If all elements of $M\sigma$ are negative literals, then $\tuple{\theta\sigma,\emptyset,H\sigma}\in\mathcal B_\tau(Q^\ast)$.
    \item Let $\sigma=\subst{T:=\tau}$ with $T$ the temporal variable in $P$.
      If every element of $M\sigma$ has timestamp strictly larger than the timestamp of $P\sigma$, then $\tuple{\theta\sigma,\emptyset,H\sigma}\in\mathcal B_\tau(Q^\ast)$.
    \end{enumerate}
  \item $\sch\tau(Q^\ast)$ is obtained from $\mathcal B_\tau(Q^\ast)$ as follows.
    Fix a topological ordering of the stratification of $\Pid$.
    There is only a finite number of $P_t$ such that $Q_P$ has at least one schematic answer with $T\theta=t$ (where $T$ is the temporal parameter in $P$).
    For each of these $P_t$ in order, set $\ell=P(t_1,\ldots,t_n)$ and let $S(P_t)$ be the set of schematic answers \tuple{\theta,E,H} to $Q_P$ such that $T\theta=t$.
    \begin{enumerate}
    \item if $S(P_t)$ contains a tuple \tuple{\theta,E,\emptyset}, then: in each $Q^\ast$, replace every \tuple{\sigma,E,H} such that $\ell\theta$ is unifiable with an element $h\in H^-$ with all possible $\tuple{\sigma\theta',E\theta',H\theta'}$ such that $\theta'$ is a minimal substitution with the property that $\ell\theta$ is not unifiable with $h\theta'$;
    \item in each $\mathcal B_\tau(Q^\ast)$, for each \tuple{\theta,E,H} such that $H^-$ contains an element $h$ with predicate symbol $P$ and timestamp $t$ and there is no tuple $\tuple{\sigma',E_\ell,H_\ell}\in S(P_t)$ such that $h$ and $\ell\sigma'$ are unifiable, then remove $\neg h$ from $H$ and add it to $E$.
    \end{enumerate}
  \end{enumerate}
\end{definition}

\begin{lemma}
  \label{lem:ground}
  The elements of $\mathcal S^\downarrow_\tau(Q^\ast)$ are exactly the ground instances of the elements of $\sch\tau(Q^\ast)$.
\end{lemma}
\begin{proof}[Proof (Sketch.)]
  Straightforward by induction on the construction of the sets $\mathcal S^\downarrow_\tau$ and $\sch\tau$, observing that the steps in Definitions~\ref{defn:Sdownarrow} and~\ref{defn:Salt} always yield sets in the relation stated in the lemma.

  For step~4, since the order of iterations of $\mathcal R$ to compute $\mathcal S^\downarrow_\tau$ from $\mathcal A_\tau$ is immaterial (Lemma~\ref{lem:R-unique}), we can assume without loss of generality that $\mathcal R$ always chooses a literal in the lowest possible stratum.
  Thus one application of $\mathcal R$ (Definition~\ref{defn:Sdownarrow}) corresponds to performing step~4 in Definition~\ref{defn:Salt} for one literal $\ell$.
  Finally, using a topological ordering guarantees that a fixpoint is reached, since processing one of the $P_t$ cannot change elements of $S(P'_t)$ for each $P'_t$ previously processed.
\end{proof}

\begin{theorem}[Complexity]
  \label{thm:algorithm-negated}
  Let $k$ be the highest arity of any predicate that occurs negated in $\Pi$.
  If $\Pi$ is $T$-stratified, then $\sch\tau$ can be computed in time exponential in $k$ and polynomial on the size of $\mathcal P_Q$, $\sch{\tau-1}$, $D|_\tau$ and the total number of queries.
\end{theorem}
\begin{proof}
  Steps~(1), (2) and (3a) are as in Theorem~\ref{thm:algorithm}, but they are now iterating over all sets of queries.
  Therefore the previous time bounds apply, multiplied by the total number of queries.
  Steps~(3b) and~(3c) can be done in time linear on the size of $\mathcal P_Q$.

  Each iteration of step~(4) requires computing the set of negative literals in a schematic answer (which can be done in time linear on the size of the answer) and checking whether a corresponding tuple exists in the current $\mathcal B_\tau$.
  This can be done in time linear in the size of $\mathcal B_\tau$, which is polynomial in the size of $\sch{\tau-1}$ (since it was computed from $\sch{\tau-1}$ by a polynomial algorithm).
  The number of substitutions $\sigma$ that need to be considered is bound by the number of constants to the power $k$, which is constant; thus any remaining answers can be added in worst-case polynomial time.
  Finally, the total number of iterations of step~(4) is at most the total number of queries multiplied by $\tau$.
\end{proof}

The only parameter on which there is an exponential dependency in this complexity bound is $k$.
In practice, this value is small, since predicates in real-life programs typically do not involve more than a few variables.
This minimizes the impact of this exponential complexity.

We can also improve on the average-time complexity by delaying the instantiation in step~4 through a number of other techniques.
One possibility is delaying the application of step~4 until there are no positive atoms in $H$, so that all negative literals are as instantiated as possible.
In particular, if all negations in $\Pi$ are safe, the only potential uninstantiated variables are those in the original query $Q$.
Another possibility is storing $\sigma$ abstractly, e.g.~as a set of constraints on the substitution $\theta$ (such as \subst{X\neq\m{a}}).
In the worst case, though, any of these approaches still has the same complexity as the simpler version presented earlier -- since we can always find an example where the number of answers to the original query is exponential on $k$.

We now illustrate this algorithm with our running example, where all cases are covered.
\begin{example}
  \label{ex:neg-online}
  Assume that there are two patients being monitored, \m{john} and \m{gus}.
  Recall that pre-processing yielded the sets 
  \begin{align*}
    \mathcal P_{Q_R} &= \{
      \tuple{\emptyset,\{\m{BCA}(X,T+2),\underline{\neg\m{GVS}(X,T+1)}\}},
      \\ & \phantom{{}=\{}
      \tuple{\emptyset,\{\neg\m{GVS}(X,T+1),\underline{\neg\m{GVS}(X,T),\neg\m{ST}(X,\m{us},T)}\}}\\
    \mathcal P_{Q_G} &= \tuple{\emptyset,\{\underline{\m{GCM}(X,T),\m{GBOL}(X,T)}\}} \\
    \mathcal P_{Q_S} &= \{\tuple{\emptyset,\{\underline{\m{BCA}(X,T+1)}\}},
      \\ & \phantom{{}=\{}
      \tuple{\emptyset,\{\underline{\neg\m{GVS}(X,T-1),\neg\m{ST}(X,\m{us},T-1)}\}}
  \end{align*}
  where the literals with minimal timestamp in each set of premises are underlined, and that, by definition,
  $\sch{-1}(Q_R)=\sch{-1}(Q_G)=\sch{-1}(Q_S)=\emptyset$.

  We start at time point $\tau=0$, and assume that $D|_0=\{\m{GCM(gus,0)}\}$.
  Since $\sch{\tau-1}(Q^\ast)=\emptyset$ for all queries $Q^\ast$, step~2 is trivial.

  The table below shows the schematic hypothetical answers added in each substep of step~3.
  Substeps~3(b) and 3(c) uses $\sigma=\subst{T:=0}$ in all queries.

  \begin{center}
  \begin{tabular}{cccc} \toprule
    Step & $\mathcal B_0(Q_R)$ & $\mathcal B_0(Q_G)$ & $\mathcal B_0(Q_S)$ \\ \midrule
    3(a)
    & nothing
    & nothing
    & nothing \\ \midrule
    3(b)
    & \tuple{\sigma,\emptyset,\{\m{BCA}(X,2),\neg\m{GVS}(X,1)\}}
    & nothing
    & nothing \\
    & \tuple{\sigma,\emptyset,\{\neg\m{GVS}(X,1),\neg\m{GVS}(X,0),\neg\m{ST}(X,\m{us},0)\}} \\ \midrule
    3(c)
    & \tuple{\sigma,\emptyset,\{\m{BCA}(X,2),\neg\m{GVS}(X,1)\}}
    & nothing
    & \tuple{\sigma,\emptyset,\{\m{BCA}(X,1)\}} \\ \bottomrule
  \end{tabular}
  \end{center}

  We obtain:
  \begin{align*}
    \mathcal B_0(Q_R)=\{ & \tuple{\subst{T:=0},\emptyset,\{\m{BCA}(X,2),\neg\m{GVS}(X,1)\}},\\
    & \tuple{\subst{T:=0},\emptyset,\{\neg\m{GVS}(X,1),\neg\m{GVS}(X,0),\neg\m{ST}(X,\m{us},0)\}}\} \\
    \mathcal B_0(Q_G)=\emptyset \\
    \mathcal B_0(Q_S)=\{ & \tuple{\subst{T:=0},\emptyset,\{\m{BCA}(X,1)\}}\}
  \end{align*}
  We move on to step~4.
  In the stratification (up to timestamp~$0$), \m{GVS_0} and \m{ST_0} at the lowest stratum, and \m{Risk_0} depends on both of them; we include only the relevant information on each $S(P_t)$ -- the substitution and whether the set of hypotheses is empty.
  Since \m{Risk_0} is at the highest stratum, no schematic hypothetical answers can depend on it, and we do not include it in the table.

  \begin{center}
  \begin{tabular}{cccc} \toprule
    $P_t$ & $S(P_t)$ & set & tuples of interest \\ \midrule
    \m{GVS_0} & $\emptyset$
    & $\mathcal B_0(Q_R)$
    & \tuple{\subst{T:=0},\emptyset,\{\neg\m{GVS}(X,1),\neg\m{GVS}(X,0),\neg\m{ST}(X,\m{us},0)\}} \\
    &&& becomes \\
    &&& \tuple{\subst{T:=0},\{\neg\m{GVS}(X,0)\},\{\neg\m{GVS}(X,1),\neg\m{ST}(X,\m{us},0)\}} \\ \midrule
    \m{ST_0} & $\subst{T:=0}$
    & $\mathcal B_0(Q_R)$
    & \tuple{\subst{T:=0},\{\neg\m{GVS}(X,0)\},\{\neg\m{GVS}(X,1),\neg\m{ST}(X,\m{us},0)\}} \\
    & $H\neq\emptyset$ && unchanged \\ \bottomrule
  \end{tabular}
  \end{center}
  
  We finally obtain:
  \begin{align*}
    \sch0(Q_R)=\{ & \tuple{\subst{T:=0},\emptyset,\{\m{BCA}(X,2),\neg\m{GVS}(X,1)\}},\\
    & \tuple{\subst{T:=0},\{\neg\m{GVS}(X,0)\},\{\neg\m{GVS}(X,1),\neg\m{ST}(X,\m{us},0)\}}\} \\
    \sch0(Q_G)=\emptyset \\
    \sch0(Q_S)=\{ & \tuple{\subst{T:=0},\emptyset,\{\m{BCA}(X,1)\}}\}
  \end{align*}

  We move on to $\tau=1$, and assume that $D|_1=\{\m{BCA(john,1)},\m{GCM(gus,1)}\}$.
  We summarize the application of step~2 in the next table.
  \begin{center}
  \begin{tabular}{cccc} \toprule
    set & tuple \\ \midrule
    $\mathcal B_1(Q_R)$
    & \tuple{\subst{T:=0},\emptyset,\{\m{BCA}(X,2),\neg\m{GVS}(X,1)\}} \\
    & (unchanged from $\sch0(Q_R)$)
    \\ \midrule
    $\mathcal B_1(Q_R)$
    & \tuple{\subst{T:=0},\{\neg\m{GVS}(X,0)\},\{\neg\m{GVS}(X,1),\neg\m{ST}(X,\m{us},0)\}}
    \\
    & (unchanged from $\sch0(Q_R)$)
    \\ \midrule
    $\mathcal B_1(Q_S)$
    & \tuple{\subst{T:=0,X:=\m{john}},\{\m{BCA(john,1)}\},\emptyset} \\
    & (from unifying $\m{BCA}(X,1)$ with $D|_1$)
    \\ \bottomrule
  \end{tabular}
  \end{center}

  Furthermore, in step~3 we obtain the following additional schematic hypothetical answers, where substeps~3(b) and 3(c) use $\sigma=\subst{T:=1}$ in all queries.

  \begin{center}\small
  \begin{tabular}{cccc} \toprule
    Step & $\mathcal B_1(Q_R)$ & $\mathcal B_1(Q_G)$ & $\mathcal B_1(Q_S)$ \\ \midrule
    3(a)
    & nothing
    & nothing
    & \tuple{\subst{T:=0,X:=\m{john}},\{\m{BCA(john,1)}\},\emptyset} \\ \midrule
    3(b)
    & \tuple{\sigma,\emptyset,\{\m{BCA}(X,3),\neg\m{GVS}(X,2)\}}
    & nothing
    & \tuple{\sigma,\emptyset,\{\neg\m{GVS}(X,0),\neg\m{ST}(X,\m{us},0)\}} \\
    & \tuple{\sigma,\emptyset,\{\neg\m{GVS}(X,2),\neg\m{GVS}(X,1),\neg\m{ST}(X,\m{us},1)\}} \\ \midrule
    3(c)
    & \tuple{\sigma,\emptyset,\{\m{BCA}(X,3),\neg\m{GVS}(X,2)\}}
    & nothing
    & \tuple{\sigma,\emptyset,\{\m{BCA}(X,2)\}} \\ \bottomrule
  \end{tabular}
  \end{center}

  Combining all these results, we obtain:
  \begin{align*}
    \mathcal B_1(Q_R) = \{
    & \tuple{\subst{T:=0},\emptyset,\{\m{BCA}(X,2),\neg\m{GVS}(X,1)\}}, \\
    & \tuple{\subst{T:=0},\{\neg\m{GVS}(X,0)\},\{\neg\m{GVS}(X,1),\neg\m{ST}(X,\m{us},0)\}}, \\
    & \tuple{\subst{T:=1},\emptyset,\{\m{BCA}(X,3),\neg\m{GVS}(X,2)\}}, \\
    & \tuple{\subst{T:=1},\emptyset,\{\neg\m{GVS}(X,2),\neg\m{GVS}(X,1),\neg\m{ST}(X,\m{us},1)\}}\} \\
    \mathcal B_1(Q_G) = \emptyset \\
    \mathcal B_1(Q_S) = \{
    & \tuple{\subst{T:=1},\emptyset,\{\neg\m{GVS}(X,0),\neg\m{ST}(X,\m{us},0)\}}, \\
    & \tuple{\subst{T:=0,X:=\m{john}},\{\m{BCA(john,1)}\},\emptyset}, \\
    & \tuple{\subst{T:=1},\emptyset,\{\m{BCA}(X,2)\}}\} \\
  \end{align*}
  and we observe that \subst{T:=0,X:=\m{john}} is an answer to $Q_S$.
  In order to perform step~4, we note that the relevant strata are: \m{GVS_0}, \m{GVS_1} and \m{ST_0} at the lowest level; \m{ST_1} at the next level; and finally \m{Risk_0} and \m{Risk_1} in the two highest levels, which we omit since no other queries depend on them.
  We again summarize application of step~4 in a table.

  \begin{center}\small
  \begin{tabular}{cccc} \toprule
    $P_t$ & $S(P_t)$ & set & tuples of interest \\ \midrule
    \m{GVS_0} & $\emptyset$
    & $\mathcal B_1(Q_S)$
    & \tuple{\subst{T:=1},\emptyset,\{\neg\m{GVS}(X,0),\neg\m{ST}(X,\m{us},0)\}} \\
    &&& becomes \\
    &&& \tuple{\subst{T:=1},\{\neg\m{GVS}(X,0)\},\{\neg\m{ST}(X,\m{us},0)\}} \\ \midrule
    \m{GVS_1} & $\emptyset$
    & $\mathcal B_1(Q_R)$
    & \tuple{\subst{T:=0},\emptyset,\{\m{BCA}(X,2),\neg\m{GVS}(X,1)\}} \\
    &&& becomes \\
    &&& \tuple{\subst{T:=0},\{\neg\m{GVS}(X,1)\},\{\m{BCA}(X,2)\}} \\ \midrule
    && $\mathcal B_1(Q_R)$
    & \tuple{\subst{T:=0},\{\neg\m{GVS}(X,0)\},\{\neg\m{GVS}(X,1),\neg\m{ST}(X,\m{us},0)\}} \\
    &&& becomes \\
    &&& \tuple{\subst{T:=0},\{\neg\m{GVS}(X,0),\neg\m{GVS}(X,1)\},\{\neg\m{ST}(X,\m{us},0)\}} \\ \midrule
    && $\mathcal B_1(Q_R)$
    & \tuple{\subst{T:=1},\emptyset,\{\neg\m{GVS}(X,2),\neg\m{GVS}(X,1),\neg\m{ST}(X,\m{us},1)\}} \\
    &&& becomes \\
    &&& \tuple{\subst{T:=1},\{\neg\m{GVS}(X,1)\},\{\neg\m{GVS}(X,2),\neg\m{ST}(X,\m{us},1)\}} \\ \midrule
    \m{ST_0} & \subst{T:=0,X:=\m{john}}
    & $\mathcal B_1(Q_S)$
    & \tuple{\subst{T:=1},\{\neg\m{GVS}(X,0)\},\{\neg\m{ST}(X,\m{us},0)\}} \\
    & $H=\emptyset$ && becomes, from 4(a), \\
    &&& \tuple{\subst{T:=1,X:=\m{gus}},\{\neg\m{GVS(gus,0)}\},\{\neg\m{ST(gus,us,0)}\}} \\
    &&& and then, from 4(b), \\
    &&& \tuple{\subst{T:=1,X:=\m{gus}},\{\neg\m{GVS(gus,0)},\neg\m{ST(gus,us,0)}\},\emptyset} \\ \midrule
    && $\mathcal B_1(Q_R)$
    & \tuple{\subst{T:=0},\{\neg\m{GVS}(X,0),\neg\m{GVS}(X,1)\},\{\neg\m{ST}(X,\m{us},0)\}} \\
    &&& becomes, from 4(a), \\
    &&& \tuple{\subst{T:=0,X:=\m{gus}},\{\neg\m{GVS}(\m{gus},0),\neg\m{GVS}(\m{gus},1)\},\{\neg\m{ST}(\m{gus},\m{us},0)\}} \\
    &&& and then, from 4(b), \\
    &&& \tuple{\subst{T:=0,X:=\m{gus}},\{\neg\m{GVS}(\m{gus},0),\neg\m{GVS}(\m{gus},1),\neg\m{ST}(\m{gus},\m{us},0)\},\emptyset} \\ \midrule
    \m{ST_1} & \subst{T:=1,X:=\m{gus}}
    & $\mathcal B_1(Q_R)$
    & \tuple{\subst{T:=1},\{\neg\m{GVS}(X,1)\},\{\neg\m{GVS}(X,2),\neg\m{ST}(X,\m{us},1)\}} \\
    & $H=\emptyset$ && becomes, from 4(a), \\
    &&& \tuple{\subst{T:=1,X:=\m{john}},\{\neg\m{GVS(john,1)}\},\{\neg\m{GVS(john,2)},\neg\m{ST(john,us,1)}\}} \\
    &&& and then, from 4(b), \\
    &&&
    \tuple{\subst{T:=1,X:=\m{john}},\{\neg\m{GVS(john,1)},\neg\m{ST(john,us,1)}\},\{\neg\m{GVS(john,2)}\}}
    \\ \bottomrule
  \end{tabular}
  \end{center}
  At the end of this step, we have:
  \begin{align*}
    \sch1(Q_R) = \{
    & \tuple{\subst{T:=0},\{\neg\m{GVS}(X,1)\},\{\m{BCA}(X,2)\}}, \\
    & \tuple{\subst{T:=0,X:=\m{gus}},\{\neg\m{GVS}(\m{gus},0),\neg\m{GVS}(\m{gus},1),\neg\m{ST}(\m{gus},\m{us},0)\},\emptyset}, \\
    & \tuple{\subst{T:=1},\emptyset,\{\m{BCA}(X,3),\neg\m{GVS}(X,2)\}}, \\
    & \tuple{\subst{T:=1,X:=\m{john}},\{\neg\m{GVS(john,1)},\neg\m{ST(john,us,1)}\},\{\neg\m{GVS(john,2)}\}} \\
    \sch1(Q_G) = \emptyset \\
    \sch1(Q_S) = \{
    & \tuple{\subst{T:=1,X:=\m{gus}},\{\neg\m{GVS(gus,0)},\neg\m{ST(gus,us,0)}\},\emptyset}, \\
    & \tuple{\subst{T:=0,X:=\m{john}},\{\m{BCA(john,1)}\},\emptyset}, \\
    & \tuple{\subst{T:=1},\emptyset,\{\m{BCA}(X,2)\}}\}
  \end{align*}

  We now consider $\tau=2$, with $D|_2=\{\m{GCM(gus,2)},\m{GBOL(gus,2)}\}$.
  In step~2, the two answers from $\sch1$ with $\m{BCA}(X,2)$ as hypothesis are ignored (since $\m{BCA}(X,2)$ does not unify with any element of $D|_2$) and the remaining ones are copied to $\mathcal B_2$.
  After this step, we have therefore that
  \begin{align*}
    \mathcal B_2(Q_R) = \{
    & \tuple{\subst{T:=0,X:=\m{gus}},\{\neg\m{GVS}(\m{gus},0),\neg\m{GVS}(\m{gus},1),\neg\m{ST}(\m{gus},\m{us},0)\},\emptyset}, \\
    & \tuple{\subst{T:=1},\emptyset,\{\m{BCA}(X,3),\neg\m{GVS}(X,2)\}}, \\
    & \tuple{\subst{T:=1,X:=\m{john}},\{\neg\m{GVS(john,1)},\neg\m{ST(john,us,1)}\},\{\neg\m{GVS(john,2)}\}} \\
    \mathcal B_2(Q_G) = \emptyset \\
    \mathcal B_2(Q_S) = \{
    & \tuple{\subst{T:=1,X:=\m{gus}},\{\neg\m{GVS(gus,0)},\neg\m{ST(gus,us,0)}\},\emptyset}, \\
    & \tuple{\subst{T:=0,X:=\m{john}},\{\m{BCA(john,1)}\},\emptyset}\}
  \end{align*}

  In step~3a we obtain the answer \tuple{\subst{T:=2,X:=\m{gus}},\{\m{GCM(gus,2)},\m{GBOL(gus,2)}\},\emptyset} in $\mathcal B_2(Q_G)$, while steps~3b and~3c are similar to the previous timestamps.
  After step~3, we have therefore
  \begin{align*}
    \mathcal B_2(Q_R) = \{
    & \tuple{\subst{T:=0,X:=\m{gus}},\{\neg\m{GVS(gus,0)},\neg\m{GVS(gus,1)},\neg\m{ST(gus,us,0)}\},\emptyset}, \\
    & \tuple{\subst{T:=1},\emptyset,\{\m{BCA}(X,3),\neg\m{GVS}(X,2)\}}, \\
    & \tuple{\subst{T:=1,X:=\m{john}},\{\neg\m{GVS(john,1)},\neg\m{ST(john,us,1)}\},\{\neg\m{GVS(john,2)}\}} \\
    & \tuple{\subst{T:=2},\emptyset,\{\m{BCA}(X,4),\neg\m{GVS}(X,3)\}}, \\
    & \tuple{\subst{T:=2},\emptyset,\{\neg\m{GVS}(X,3),\neg\m{GVS}(X,2),\neg\m{ST}(X,\m{us},2)\}} \\
    \mathcal B_2(Q_G) = \{
    & \tuple{\subst{T:=2,X:=\m{gus}},\{\m{GCM(gus,2)},\m{GBOL(gus,2)}\},\emptyset}\} \\
    \mathcal B_2(Q_S) = \{
    & \tuple{\subst{T:=1,X:=\m{gus}},\{\neg\m{GVS(gus,0)},\neg\m{ST(gus,us,0)}\},\emptyset}, \\
    & \tuple{\subst{T:=0,X:=\m{john}},\{\m{BCA(john,1)}\},\emptyset}, \\
    & \tuple{\subst{T:=2},\emptyset,\{\neg\m{GVS}(X,1),\neg\m{ST}(X,\m{us},1)\}}, \\
    & \tuple{\subst{T:=2},\emptyset,\{\m{BCA}(X,3)\}}\} \\
  \end{align*}

  The relevant strata for step~4 are now \m{GVS_1}, \m{GVS_2} and \m{ST_1} at the lowest level, followed by \m{ST_2} at the next level.
  As before we can omit \m{Risk} from the discussion, since no other queries depend on that predicate.
  This step is summarized in the table below.

  \begin{center}\small
  \begin{tabular}{cccc} \toprule
    $P_t$ & $S(P_t)$ & set & tuples of interest \\ \midrule
    \m{GVS_1} & $\emptyset$
    & $\mathcal B_2(Q_S)$
    & \tuple{\subst{T:=2},\emptyset,\{\neg\m{GVS}(X,1),\neg\m{ST}(X,\m{us},1)\}} \\
    &&& becomes \\
    &&& \tuple{\subst{T:=2},\{\neg\m{GVS}(X,1)\},\{\neg\m{ST}(X,\m{us},1)\}} \\ \midrule
    \m{GVS_2} & \subst{T:=2,X:=\m{gus}}
    & $\mathcal B_2(Q_R)$
    & \tuple{\subst{T:=1},\emptyset,\{\m{BCA}(X,3),\neg\m{GVS}(X,2)\}} \\
    & $H=\emptyset$ && becomes, from 4(a) and 4(b), \\
    &&& \tuple{\subst{T:=1,X:=\m{john}},\{\neg\m{GVS(john,2)}\},\{\m{BCA(john,3)}\}} \\ \midrule
    && $\mathcal B_2(Q_R)$
    & \tuple{\subst{T:=1,X=\m{john}},\{\neg\m{GVS(john,1)},\neg\m{ST(john,us,1)}\},\{\neg\m{GVS(john,2)}\}} \\
    &&& becomes \\
    &&&\tuple{\subst{T:=1,X=\m{john}},\{\neg\m{GVS(john,1)},\neg\m{ST(john,us,1)},\neg\m{GVS(john,2)}\},\emptyset} \\ \midrule
    && $\mathcal B_2(Q_R)$
    & \tuple{\subst{T:=2},\emptyset,\{\neg\m{GVS}(X,3),\neg\m{GVS}(X,2),\neg\m{ST}(X,\m{us},2)\}} \\
    &&& becomes, from 4(a) and 4(b), \\
    &&& \tuple{\subst{T:=2,X:=\m{john}},\{\neg\m{GVS}(\m{john},2)\},\{\neg\m{GVS}(\m{john},3),\neg\m{ST}(\m{john},\m{us},2)\}} \\ \midrule
    \m{ST_1} & \subst{T:=1,X:=\m{gus}}
    & $\mathcal B_2(Q_S)$
    & \tuple{\subst{T:=2},\{\neg\m{GVS}(X,1)\},\{\neg\m{ST}(X,\m{us},1)\}} \\
    & $H=\emptyset$ && becomes, from 4(a) and 4(b), \\
    &&& \tuple{\subst{T:=2,X:=\m{john}},\{\neg\m{GVS(john,1)},\neg\m{ST(john,us,1)}\},\emptyset}
    \\ \midrule
    \m{ST_2} & \subst{T:=2,X:=\m{john}}
    & $\mathcal B_2(Q_R)$
    & \tuple{\subst{T:=2,X:=\m{john}},\{\neg\m{GVS}(\m{john},2)\},\{\neg\m{GVS}(\m{john},3),\neg\m{ST}(\m{john},\m{us},2)\}} \\
    & $H=\emptyset$ && is removed in step 4(a)
    \\ \midrule
  \end{tabular}
  \end{center}
  At the end of this step, we have:
  \begin{align*}
    \sch2(Q_R) = \{
    & \tuple{\subst{T:=0,X:=\m{gus}},\{\neg\m{GVS(gus,0)},\neg\m{GVS(gus,1)},\neg\m{ST(gus,us,0)}\},\emptyset}, \\
    & \tuple{\subst{T:=1,X:=\m{john}},\{\neg\m{GVS(john,2)}\},\{\m{BCA(john,3)}\}}, \\
    & \tuple{\subst{T:=1,X:=\m{john}},\{\neg\m{GVS(john,1)},\neg\m{ST(john,us,1)},\neg\m{GVS(john,2)}\},\emptyset}\} \\
    & \tuple{\subst{T:=2},\emptyset,\{\m{BCA}(X,4),\neg\m{GVS}(X,3)\}} \\
    \sch2(Q_G) = \{
    & \tuple{\subst{T:=2,X:=\m{gus}},\{\m{GCM(gus,2)},\m{GBOL(gus,2)}\},\emptyset}\} \\
    \sch2(Q_S) = \{
    & \tuple{\subst{T:=1,X:=\m{gus}},\{\neg\m{GVS(gus,0)},\neg\m{ST(gus,us,0)}\},\emptyset}, \\
    & \tuple{\subst{T:=0,X:=\m{john}},\{\m{BCA(john,1)}\},\emptyset}, \\
    & \tuple{\subst{T:=2,X:=\m{john}},\{\neg\m{GVS(john,1)},\neg\m{ST(john,us,1)}\},\emptyset}, \\
    & \tuple{\subst{T:=2},\emptyset,\{\m{BCA}(X,3)\}}\}
  \end{align*}
\eoe
\end{example}

\paragraph{Forgetting} The previous example shows that the online algorithm needs to keep track of a potentially large number of schematic supported answers while running.
Furthermore, even answers without premises need to be kept due to the first check in step~3.
This may raise concerns about this algorithm having potentially unbounded memory requirements.

For stratified programs, this is not the case, since all schematic answers eventually become answers (with no premises) or are discarded (this follows by combining Proposition~\ref{lem:negated-compl} with Lemma~\ref{lem:ground}).
Furthermore, answers can also be forgotten eventually: static analysis on the sets $\mathcal P_Q$ can be used to determine how long an answer can influence the results of the algorithm (essentially, this is a refinement of the notion of the delay of the query), so that all answers can be discarded after enough time has passed.
Hypothetical answers can also be discarded when there is an answer with the same substitution (this is the case for the query $Q_R$ and \subst{T:=1,X:=\m{john}} at the end of the last example).

\paragraph{Explicit instantiation} Step~4 requires instantiating answers with all constants that do not satisfy a particular condition, which can immensely increase the number of schematic answers that have to be considered.
In a practical implementation, as mentioned earlier, we expect that this can be addressed by adding constraints to the substitutions in these answers instead of explicitly instantiating them; in our example, in the different applications of step~4(a), we would typically add constraints $X\neq\m{john}$ or $X\neq\m{gus}$ instead of explicitly instantiating $X$.
In this manner, we expect that we can partially avoid the exponential blowup that this step of the algorithm might cause.

\section{Related work and discussion}
\label{sec:rw}

Our work contributes to the field of stream reasoning, the task of conjunctively reasoning over streaming data and background knowledge~\cite{DellaValle2008,DellaValle2009}.

Research advances on Complex Event Processors and Data Stream Management Systems~\cite{Cugola2012}, together with Knowledge Representation and the Semantic Web, all contributed to the several stream reasoning languages, systems and mechanisms proposed during the last decade~\cite{DellAglio2017}.

Hypothetical query answering over streams can be broadly related to works that study abduction in logic programming~\cite{Inoue94,Denecker02}, namely those that view negation in logic programs as hypotheses and relate it to contradiction avoidance~\cite{Alferes96,Dung91}.
Furthermore, our framework can be characterized as applying an incremental form of data-driven abductive inference.
Such forms of abduction have been used in other works~\cite{Meadows13}, but with a rather different approach and in the plan interpretation context.
To our knowledge, hypothetical or abductive reasoning has not been previously used in the context of stream reasoning, to answer continuous queries, although it has been applied to annotation of stream data~\cite{Alirezaie14}.

Below we focus on four dimensions of stream reasoning — incremental evaluation, unbound wait and blocking queries, incomplete information, and negation —, which, in our opinion, due to the characteristics of our work, provide the most interesting and meaningful comparison with existing works in the field.

\paragraph{Incremental evaluation}
Computing answers to a query over a data source that is continuously producing information, be it at slow or very fast rates, asks for techniques that allow for some kind of \emph{incremental evaluation}, in order to avoid reevaluating the query from scratch each time a new tuple of information arrives.
Several efforts have been made in that direction, capitalising on incremental algorithms based on seminaive evaluation~\cite{Abiteboul1995,Barbieri2010,Gupta1993,Hu2018,Motik2015}, on truth maintenance systems~\cite{Beck2015}, window oriented~\cite{Ghanem2007} among others.
Our algorithm fits naturally in the first class, as it is an incremental variant of SLD-resolution.

\paragraph{Unbound wait and blocking queries}
The problems of unbound wait and blocking queries have deserved much attention in the area of query answering over data streams.
There have been efforts to identify the problematic issues~\cite{Law2011} and varied proposals to cope with their negative effects, as in~\cite{Becketal2015,Ozcep2017,Ronca2018,Zaniolo2012} among others.
Our framework deals with unbound wait by outputting at each time point all supported answers (including some that later may prove to be false), as illustrated e.g.~in Example~\ref{ex:good}.
Furthermore, if we receive a supported answer whose time parameters are all instantiated, then we immediately have a bound on how long we have to wait until the answer is definite (or rejected).

Blocking queries may still be a problem in our framework, though: as seen in Example~\ref{ex:infinite}, blocking operations (in the form of infinitely recursive predicates) may lead to infinite SLD-derivations, which cause the pre-processing step of our algorithm to diverge.
We showed that syntactic restrictions of the kind already considered by~\cite{Motik} are guaranteed to avoid blocking queries; \begin{newtext}Example~\ref{ex:unbound} illustrates that, as in the cited work, these conditions are however not necessary: pre-processing may terminate even though they do not hold.\end{newtext}

Our algorithm also implicitly embodies a forgetting algorithm, as the only elements of $D_\tau$ kept in the $E$ sets are those that are still relevant to compute future answers to $Q$.

The approaches from~\cite{Ronca2018,Walega2019,Zaniolo2012}, which are also specifically Datalog-based, make some restrictions to the form of rules in order to avoid blocking behaviours.

In \cite{Zaniolo2012}, working in a variant of Datalog called Streamlog, the author characterises \emph{sequential rules and programs} with the purpose of avoiding blocking behaviour.
In particular, the rule $\m{Shdn}(X,T) \rightarrow \m{Malf}(X,T-2)$ from Example~\ref{ex:toy} is disallowed in this framework, since the timestamp in the head of a rule may never be smaller than the timestamps of the atoms in the body.

The authors of \cite{Ronca2018} propose the \emph{language of forward-propagating queries}, yet another variant of Temporal Datalog that allows queries to be answered in polynomial time in the size of the input data.
This is again achieved at the cost of disallowing propagation of derived facts towards past time points -- once more precluding rule $\m{Shdn}(X,T) \rightarrow \m{Malf}(X,T-2)$ in Example~\ref{ex:toy}.
The authors present a generic algorithm to compute answers to a query, given a window size, and investigate methods for calculating a minimal window size.

Memory consumption is a concern in~\cite{Walega2019}, where a sound and complete stream reasoning algorithm is presented for a fragment of datalogMTL -- forward-propagating datalogMTL -- that also disallows propagation of derived information towards the past.
DatalogMTL is a Datalog extension of the Horn fragment of MTL~\cite{Koymans1990,Alur1993}, which was proposed in~\cite{Brandt2018} for ontology-based access to temporal log data.
DatalogMTL rules allow propagation of derived information to both past and future time points.
Concerned with the possibly huge or unbounded set of input facts that have to be kept in memory, the authors of~\cite{Walega2019} restrict the set of operators of datalogMTL to one that generates only so-called forward-propagating rules.
They present a sound and complete algorithm for stream reasoning with forward-propagating datalogMTL that limits memory usage through the use of a sliding window.

\begin{newtext}
More recently, \cite{Ronca2022} removes these syntactic restrictions on the forms of rules (keeping only the light requirement of temporal guardedness, which is similar to connectedness) and studies fragments of the language for which the existence of a query delay is decidable.
\end{newtext}

\paragraph{Incomplete information}
When reasoning over sources with incomplete information, the concepts of certain and possible answers, and granularities thereof, inevitably arise~\cite{deLeng2019,Grayetal2007,Lang2014,Razniewski2015} as a way to assign confidence levels to the information output to the user.
These approaches, like ours, also compute answers with incomplete information.

However, our proposal is substantially different from those works, since they focus on answers over past incomplete information.
First, as mentioned in Section~\ref{sec:backgr}, we assume that our time stream is complete, in the sense that whenever it produces a fact about a time instant $\tau$, all EDB facts about time instants $\tau'<\tau$ are already there (in line with the progressive closed world assumption of \cite{Zaniolo2012}).
Secondly, our hypothetical and supported answers are built over present facts and future, still unknown, hypotheses, and eventually either become effective answers or are discarded; in the scenario of past incomplete information, the confidence level of answers may never change.

\paragraph{Negation}
Other authors have considered languages using negation.
Our definition of stratification is similar to the concept of temporal stratification from~\cite{Zaniolo2015}.
However, temporal stratification requires the strata to be also ordered according to time (i.e., if the definition of $p_t$ is in $\Pi_i$ and the definition of $p_{t+1}$ is in $\Pi_j$, then $j\geq i$).
We make no such assumption in this work.

Beck et al.~\cite{Beck2015} also consider a notion of temporal stratification for stream reasoning.
However, their framework also includes explicit temporal operators, making the whole formalization more complex.

Insofar as we are aware, this is also the first framework for continuous query answering that does not require negation to be safe.

\section{Conclusions and future work}
\label{sec:concl}

We introduced a formalism for dealing with answers to continuous queries that are consistent with the information that is available, but may still need further confirmation from future data.
We defined the notions of hypothetical and supported answers declaratively, and showed how they can be computed by means of a variant of SLD-resolution.
By refining this idea, we designed a two-phase algorithm whose online component maintains and updates the set of supported answers.
Our framework also allows for negation, whose declarative semantics is in line with the progressive closed world assumption.

Our methodology avoids some of the typical problems with delay in continuous query answering.
In particular, hypothetical answers allow us to detect that we are not in a situation of unbound wait: if we receive a supported answer whose time parameters are all instantiated, then we immediately have a bound on how long we have to wait until the answer is definite (or rejected).
The usefulness of this information is of course application specific.
In the situation of Example~\ref{ex:blahblah}, we could for example take extra preventive measures to ensure that $\m{Temp(wt25,high,2)}$ does not become true.

The offline step of our two-phase algorithm may diverge due to issues related to blocking queries.
We showed that adequate syntactic restrictions prevent this situation from arising, but that they are not necessary.

For nonrecursive and connected queries without negation, the online step of our algorithm runs in polynomial time.
However, it may still be computationally heavy if the sets of hypothetical answers that it needs to compute are large.
\begin{newtext}In particular, pre-processing may produce output that is exponential in the size of the original program.\end{newtext}
We minimize the impact of this by representing hypothetical answers schematically (since we keep variables uninstantiated), thus limiting the space and time requirements for this algorithm.
By adding instantiation constraints to this representation we can also reduce the impact of the exponential blowup introduced by negation.
These are informal observations though, which we plan to measure more precisely in a future practical evaluation.

The sets of evidence for hypothetical answers can be used to define a partial order of relative ``confidence'' of each hypothetical answer.
It would be interesting to explore the connection between such notions of confidence and other frameworks of reasoning with multiple truth values.

We also plan to implement our algorithm and evaluate its performance in practice
\begin{newtext}-- a first step in this direction, only for the positive fragment of the language, is described in~\cite{Jonas}.\end{newtext}
It would also be interesting to explore approximation techniques for online algorithms in order to allow for unsafe forgetting, with bounds on how many wrong answers may be produced, and to extend our formalism to allow for communication delays, which are present in many interesting practical scenarios.

\section*{Acknowledgements}
This work was partially supported by the Independent Research Fund Denmark, grant DFF-7014-00041, and by FCT through the LASIGE Research Unit, ref.\ UIDB/00408/2020 and ref.\ UIDP/00408/2020.

\bibliographystyle{plain}
\bibliography{bibl}

\end{document}